\documentclass[a4paper,11pt,reqno]{amsart}
\usepackage[T1]{fontenc}
\usepackage[utf8]{inputenc}
\usepackage{lmodern}
\usepackage{microtype}
\usepackage{amsmath,amssymb,amsthm,mathtools}
\usepackage{mathrsfs}
\usepackage{tikz-cd}
\usepackage{geometry}
\usepackage{enumitem}
\usepackage[dvipsnames]{xcolor}
\usepackage{hyperref}
\hypersetup{
  colorlinks=true,
  citecolor=ForestGreen,
  urlcolor=cyan,
  linkcolor=blue,
  linktoc=page
}
\numberwithin{equation}{section}
\newtheorem{theorem}{Theorem}[section]
\newtheorem{proposition}[theorem]{Proposition}
\newtheorem{corollary}[theorem]{Corollary}
\newtheorem{lemma}[theorem]{Lemma}
\theoremstyle{remark}
\newtheorem{remark}[theorem]{Remark}
\theoremstyle{definition}
\newtheorem{definition}[theorem]{Definition}
\newtheorem{example}[theorem]{Example}
\newcommand{\Q}{\mathbb{Q}}
\newcommand{\Z}{\mathbb{Z}}
\newcommand{\C}{\mathbb{C}}
\newcommand{\F}{\mathbb{F}}
\newcommand{\Tr}{\operatorname{Tr}}
\newcommand{\Nm}{\operatorname{N}}
\newcommand{\HF}{\widehat{\mathrm{HFK}}}

\newcommand{\dd}{\mathrm{d}}

\title[Alexander torsion and Floer comparisons]{Alexander Torsion and Floer Comparisons in the BV--BFV Formalism}
\author{Nima Moshayedi}
\address{Institut f\"ur Mathematik, Universit\"at Z\"urich,
Winterthurerstrasse 190, CH-8057 Z\"urich, Switzerland}
\email{nima.moshayedi@math.uzh.ch}
\date{14 October 2026}
\begin{document}
\begin{abstract}
We compare relative Abelian BF torsion with the Alexander-graded Euler
characteristic of knot Floer homology.  A meridional annulus boundary
condition identifies the normalized relative torsion with the Alexander
polynomial.  One exact prime-cyclotomic torsion value recovers all its
coefficients under a sharp genus bound, compatibly with connected-sum
gluing and branched-cover formulas.  We compute interface contractions
and obstructions to scalar comparisons, including the pairing defect of
two-sector reduction for the degree-zero torus algebra.  A six-generator
dg resolution yields an equivalence of perfect-module categories
preserving derived pairings and admits a classical one-dimensional
BV--BFV realization.  In the fixed nonresonant polynomial quantum BFV
sector, over characteristic zero with $\hbar$ inverted, both states and
Weyl observables reduce to a vacuum field; a Hochschild invariant
obstructs a Morita equivalence with the degree-zero torus algebra.
These results do not establish a full geometric BV--BFV/Floer equivalence.
\end{abstract}
\maketitle
\tableofcontents

\section{Introduction}
Let $K\subset S^3$ be an oriented knot, $X_K=S^3\setminus\operatorname{int}N(K)$ its exterior, and
\begin{equation}\label{eq:alex}
 \Delta_K(t)\in\Z[t,t^{-1}],\qquad \Delta_K(t^{-1})=\Delta_K(t),\quad \Delta_K(1)=1,
\end{equation}
the symmetrically normalized Alexander polynomial. Three different settings produce the polynomial. Fox calculus and Reidemeister torsion recover $\Delta_K/(t-1)$ from the rank-one local systems on $X_K$; Abelian BF theory is a field-theoretic realization of corresponding torsion data; and the work of Ozsv\'ath and Szab\'o identifies $\Delta_K$ with the Alexander-graded Euler characteristic of knot Floer homology \cite{MilnorTorsion,CattaneoBF,OzsvathSzaboKnots}.

Friedl, Juh\'asz and Rasmussen identify sutured Floer Euler
characteristics with torsion \cite{FJR}; Hom, Lidman and Watson
categorify Donaldson's Alexander-module TQFT using bordered Floer theory
\cite{HLW}.  Cellular BF theory and its gluing formalism are developed
by Cattaneo, Mnev and Reshetikhin \cite{CMR}.  These results provide the
starting point for the comparisons below.

We first ask which Floer Euler characteristics can be recovered from a
single twisted torsion value.  Evaluation at a primitive prime-order
root of unity retains all Galois conjugates when the value is known
exactly as an algebraic number.  The genus bound and the normalization
$\Delta_K(1)=1$ then give a sharp coefficient-reconstruction formula.
The corresponding local system is acyclic on the exterior and its
boundary.

The absolute torsion is $\Delta_K(t)/(t-1)$, whereas the Floer Euler
characteristic is $\Delta_K(t)$.  The factor $t-1$ has a geometric
interpretation: it removes the torsion contribution of a meridional
annulus.  Relative cellular BF theory with this boundary condition
therefore produces the Alexander polynomial directly, with the usual
ambiguity from cellular lifts.  The comparison concerns based complexes
and determinant lines.

\begin{theorem}[Meridional relative BF torsion and Floer decategorification]
\label{thm:relative-intro}
Let $(X_K,\gamma_\mu)$ denote the sutured complement of a knot, with two
oppositely oriented meridional sutures, and let $R_-\subset\partial X_K$
be its negative annulus.  For $z\in\C^\times\setminus\{1\}$ satisfying
$\Delta_K(z)\ne0$, let $L_z$ denote meridional holonomy $z$.
Then the twisted cochain complexes of $X_K$, $R_-$, the whole boundary,
and the relative pair $(X_K,R_-)$ are acyclic.  With compatible cellular
lifts, the following equality holds modulo $\pm z^{\mathbb Z}$:
\begin{equation}\label{eq:relative-intro}
 \tau(X_K,R_-;L_z)
 =\frac{\tau(X_K;L_z)}{\tau(R_-;L_z)}
 \doteq\Delta_K(z).
\end{equation}
It is the determinant-line invariant of a finite-dimensional, linear
Abelian BF system with complementary relative boundary conditions.
After fixing the Alexander symmetric normalization, its Laurent polynomial
agrees with both the sutured Floer Euler-characteristic torsion of
$(X_K,\gamma_\mu)$ and
$\sum_a\chi(\widehat{HFK}(K,a))t^a$.
\end{theorem}

The terms ``determinant-line invariant'' and ``relative BF system'' are
spelled out in Sections~\ref{sec:BVBFV}--\ref{sec:sutured}.  The qualifier
\emph{modulo units} is essential: the rank-one system $L_z$ is generally
non-unimodular, so lifts, Euler structures, and normalizations cannot be
suppressed.  The established cellular BV--BFV gluing theorem of
\cite{CMR,CMRPerturbative} is not invoked without checking its coefficient
hypotheses.

\begin{theorem}[Cyclotomic BF--Floer reconstruction]\label{thm:main}
Let $K\subset S^3$ be a knot with $g(K)\leq g$, and let $p\geq 2g+1$ be an odd prime. Write $\zeta=e^{2\pi i/p}$ and denote by $L_\zeta$ the rank-one complex local system sending a positively oriented meridian of $K$ to $\zeta$. Then:
\begin{enumerate}[label=\textup{(\roman*)},leftmargin=2em]
\item $H^\bullet(X_K;L_\zeta)=H^\bullet(\partial X_K;L_\zeta)=0$, so the specialized Reidemeister torsion is defined and nonzero.
\item Fix the symmetric torsion normalization
\begin{equation}\label{eq:tau-def}
 \tau_K^{\mathrm{sym}}(t)=\frac{\Delta_K(t)}{t-1}.
\end{equation}
For every $a\in\Z$ with $-g\leq a\leq g$, one has
\begin{equation}\label{eq:main}
 \displaystyle
 \chi\bigl(\HF(K,a)\bigr)
 =\frac{1+\Tr_{\Q(\zeta)/\Q}\!\left(\zeta^{-a}(\zeta-1)\tau_K^{\mathrm{sym}}(\zeta)\right)}{p}.
\end{equation}
\item The entire symmetrically normalized Alexander polynomial, and hence every Alexander-graded Floer Euler characteristic, is uniquely determined by the single \emph{exact algebraic value} $\tau_K^{\mathrm{sym}}(\zeta)$ and the genus bound $g$.
\end{enumerate}
\end{theorem}

A second result shows that the bound on the prime cannot generally be lowered.
\begin{theorem}[Sharpness]\label{thm:sharp}
For every odd prime $p$ and every integer $g\geq(p+1)/2$, there is a fibered knot $J_p\subset S^3$ of genus $(p+1)/2$ such that
\[
 \tau_{J_p}^{\mathrm{sym}}(\zeta)=\tau_U^{\mathrm{sym}}(\zeta)=\frac{1}{\zeta-1},
\]
where $U$ denotes the unknot. The Alexander polynomials of $J_p$ and $U$ are different. Consequently, prime-order torsion evaluation is not injective on the class of knots of genus at most $g$ when $p<2g+1$.
\end{theorem}

Our reconstruction has a compatible multiplicative form. If $K=K_1\#K_2$, then the torsion of the meridional annulus corrects the product of torsions by exactly one factor of $t-1$. Under the arithmetic reconstruction this becomes the convolution of graded Floer Euler characteristics. We also recover classical branched-cover homology from a cyclotomic norm.

Sections~\ref{sec:chaininterface}--\ref{sec:lspace} compare linear
BV--BFV interface gluing with Floer tensor pairings after
 decategorification and apply the staircase classification to positive
L-space knots.  The standard twisted BF complex is acyclic, so it
cannot be quasi-isomorphic to the associated graded knot Floer complex.

Section~\ref{sec:categorical} studies derived scalar extension for the
degree-zero torus algebra.  Its projective pairing matrix
$\left(\begin{smallmatrix}2&3\\1&2\end{smallmatrix}\right)$ obstructs
independent single-field tensor factorizations.  Reduction retaining
both idempotents is conservative on perfect modules, but its pairing
defect is measured by the chord ideal and a three-layer filtration.
Section~\ref{sec:constructive} resolves the torus algebra by six
generators and realizes its representation differential in a classical
one-dimensional cotangent BV--BFV model.  The resulting equivalence is
algebraic.  Section~\ref{sec:quantum-boundary} computes a fixed
nonresonant quantum sector of Abelian BFV theory and identifies a
Hochschild obstruction to a Morita comparison with the degree-zero
torus algebra.  A geometric comparison of boundary categories and
manifold modules requires further constructions.

\section{Twisted complexes, torsion, and Abelian BF theory}\label{sec:torsion}
\subsection{Local systems and the Fox complex}
The canonical abelianization
\[
 \phi:\pi_1(X_K)\longrightarrow H_1(X_K;\Z)\cong\Z
\]
sends every oriented Wirtinger meridian to $1$. Write $\Lambda=\Z[t,t^{-1}]$, $F=\Q(t)$, and regard $F$ as a $\Z[\pi_1(X_K)]$-module through $\gamma\mapsto t^{\phi(\gamma)}$. An associated Wirtinger presentation
\[
 \pi_1(X_K)=\langle x_1,\ldots,x_n\mid r_1,\ldots,r_{n-1}\rangle
\]
gives, after abelianized Fox differentiation, the finite based complex
\begin{equation}\label{eq:fox-complex}
 0\longrightarrow F^{n-1}\xrightarrow{\,d_2\,}F^n\xrightarrow{\,d_1\,}F\longrightarrow0,
\qquad d_1=(t-1,\ldots,t-1).
\end{equation}
Here the matrix of $d_2$ is the transpose of the $(n-1)\times n$ Fox matrix $\bigl(\phi(\partial r_i/\partial x_j)\bigr)_{ij}$. It follows from the Fox identity that $d_1d_2=0$.

For a column $j$ let $A_j(t)$ denote the square matrix obtained by removing column $j$ of the Fox matrix. The Alexander--Fox theorem says that
\begin{equation}\label{eq:fox-minor}
 \det A_j(t)=\varepsilon t^m\Delta_K(t),\qquad \varepsilon\in\{\pm1\},\quad m\in\Z,
\end{equation}
with the usual choice of Wirtinger presentation. Over $F$, the complex \eqref{eq:fox-complex} is acyclic.

\begin{definition}[Acyclic chain torsion: our convention]
\label{def:milnor}
Let $0\to C_N\xrightarrow{\partial_N}\cdots
\xrightarrow{\partial_1}C_0\to0$ be an acyclic complex of
finite-dimensional vector spaces over a field, each equipped with a
specified ordered basis $c_i$.  Write
$B_i=\operatorname{im}(\partial_{i+1})$, choose an ordered basis $b_i$
of $B_i$, and lift $b_{i-1}$ to a tuple $\widetilde b_{i-1}$ in
$C_i$.  If $[v/c_i]$ is the determinant of the transition matrix
from $c_i$ to the concatenated basis $v$, define
\begin{equation}\label{eq:milnor-definition}
 \tau(C_*,c_*)=
 \prod_{i=0}^{N}
 \bigl[b_i\,\widetilde b_{i-1}/c_i\bigr]^{(-1)^{i+1}},
 \qquad b_{-1}=b_N=\varnothing.
\end{equation}
This scalar is independent of the chosen boundary bases and lifts.
In particular, the two-term complex
$0\to F\xrightarrow{\,u\,}F\to0$ in degrees $1,0$
has torsion $u^{-1}$.
\end{definition}
\begin{proof}[Verification of independence]
Changing $b_i$ by a matrix $M$ changes the determinant in degree
$i$ by $\det M$ and that in degree $i+1$ by the same factor;
the exponents in \eqref{eq:milnor-definition} are opposite, so they
cancel.  Changing a lift by an element of $B_i$ amounts to a
block-triangular basis change with diagonal blocks equal to the
identity, and leaves the determinants unchanged.
\end{proof}

\begin{proposition}[Torsion normalization]\label{prop:torsion}
For the acyclic complex \eqref{eq:fox-complex}, the Milnor torsion convention in which the circle $0\to F\xrightarrow{t-1}F\to0$ has torsion $(t-1)^{-1}$ gives
\begin{equation}\label{eq:torsion-raw}
 \tau(X_K;F_\phi)\doteq\frac{\Delta_K(t)}{t-1},
\end{equation}
where $\doteq$ means equality up to $\pm t^m$. Equation \eqref{eq:tau-def} selects the unique symmetrized numerator with value $1$ at $t=1$.
\end{proposition}
\begin{proof}
Choose $j=n$ for simplicity. A basis for the image of $d_1$ is $(t-1)e_0$ and a lift is $e_n$ in $F^n$. The columns of $d_2$ together with $e_n$ form a basis of $F^n$, whose determinant is $\pm\det A_n(t)$. In the stated alternating-basis convention the torsion is their determinant divided by $t-1$, up to sign. Equation \eqref{eq:fox-minor} finishes the calculation.
\end{proof}

\begin{remark}\label{rem:conventions}
A gauge-fixed Gaussian BF integral may produce the torsion, its inverse, or a density in a determinant line, depending on field parity, the chosen polarization and the integration convention. Throughout this paper $\tau_K^{\mathrm{sym}}$ means the \emph{explicitly defined torsion expression} \eqref{eq:tau-def}; it is not asserted to be a convention-independent partition function. The factor $t-1$ is indispensable, including in connected-sum formulas.
\end{remark}

\subsection{The twisted Abelian BF model}
For a complex character $z\in\C^\times$, let $L_z$ be the flat line associated to $\gamma\mapsto z^{\phi(\gamma)}$ and $L_z^\vee$ its dual. A formal twisted three-dimensional Abelian BF action is
\begin{equation}\label{eq:bf-action}
 S_z[A,B]=\int_{X_K}\langle B\wedge \dd_{L_z}A\rangle,
 \qquad A\in\Omega^1(X_K;L_z),\quad B\in\Omega^1(X_K;L_z^\vee).
\end{equation}
The infinitesimal bulk gauge symmetries are $A\mapsto A+\dd_{L_z}\lambda$ and $B\mapsto B+\dd_{L_z^\vee}\eta$, with boundary terms requiring a boundary condition or a BFV extension. Formally, the variation includes
\[
 \delta S_z\big|_{\partial X_K}=\int_{\partial X_K}\langle B\wedge\delta A\rangle,
\]
which supplies the canonical boundary one-form of the BFV description. The BV resolution uses the twisted de Rham complex and its dual. Upon cellular reduction and a choice of gauge fixing, determinant data of these complexes give torsion-type amplitudes, as in the cellular BF approach of \cite{CMR}. This background motivates our terminology \emph{BF torsion sector}; it is not used as an unproved functional-integral identity.

\begin{lemma}[Acyclic boundary sector]\label{lem:boundary}
If $z\ne1$, then $H^\bullet(\partial X_K;L_z)=0$.
\end{lemma}
\begin{proof}
Choose meridian $\mu$ and preferred longitude $\lambda$ on $\partial X_K\cong T^2$. Their holonomies are $z$ and $1$, respectively. The twisted cellular cochain complex of the torus has the form
\[
 0\longrightarrow\C\xrightarrow{\binom{z-1}{0}}\C^2
 \xrightarrow{\begin{psmallmatrix}0&z-1\end{psmallmatrix}}\C\longrightarrow0
\]
(up to a harmless sign in the last differential). For $z\ne1$ it is exact.
\end{proof}

\begin{lemma}[Prime characters are nonresonant]\label{lem:nonvanish}
For every knot $K$ and every primitive $p$th root $\zeta$ of unity, with $p$ prime, one has $\Delta_K(\zeta)\ne0$. Consequently $H^\bullet(X_K;L_\zeta)=0$.
\end{lemma}
\begin{proof}
If $\Delta_K(\zeta)=0$, then the minimal polynomial $\Phi_p(t)=1+t+\cdots+t^{p-1}$ divides $t^N\Delta_K(t)\in\Z[t]$ for a sufficiently large $N$, by Gauss's lemma. Evaluation at $t=1$ gives $p\mid\Delta_K(1)=1$, a contradiction. Thus an Alexander--Fox minor is invertible at $t=\zeta$, and the specialized complex \eqref{eq:fox-complex} is exact. The knot exterior has the homotopy type of the corresponding two-dimensional spine.
\end{proof}

\section{BV--BFV boundary data and the meridional relative model}
\label{sec:BVBFV}
\subsection{Boundary data}
For clarity, we first work with smooth forms and then with finite cellular
cochains.  On an oriented three-manifold $X$ with a flat rank-one system
$L_z$, introduce the total (ghost-extended) fields
\begin{equation}\label{eq:superfields}
 \mathcal A\in\Omega^\bullet(X;L_z)[1],\qquad
 \mathcal B\in\Omega^\bullet(X;L_z^\vee)[1].
\end{equation}
A $k$-form component of $\mathcal A$ or $\mathcal B$ has ghost number
$1-k$; the degree-zero components are the classical $1$-form fields
$A$ and $B$.  The bracket $[1]$ records the degree shift, not a tensor
power.  The fields include ghosts, antifields and their differential-form
partners.  The natural pairings use
$L_z\otimes L_z^\vee\to\underline{\mathbb C}$.

Up to the standard Koszul sign convention, the odd symplectic form,
action, and cohomological vector field are
\begin{align}
 \omega_X&=\int_X\langle\delta\mathcal B\wedge\delta\mathcal A\rangle,
 &\deg\omega_X&=-1,\label{eq:BVsymp}\\
 S_X&=\int_X\langle\mathcal B\wedge d_z\mathcal A\rangle,
 &Q_X\mathcal A&=d_z\mathcal A,\quad
 Q_X\mathcal B=d_{z^{-1}}\mathcal B.\label{eq:BVaction}
\end{align}
Here $d_z^2=d_{z^{-1}}^2=0$ because the background connections are flat.
An integration by parts in $\delta S_X$ creates the boundary term
\begin{equation}\label{eq:alphaBFV}
 \alpha_{\partial X}=\int_{\partial X}
 \langle\mathcal B_{\partial}\wedge\delta\mathcal A_{\partial}\rangle,
 \qquad \omega_{\partial X}=\delta\alpha_{\partial X}.
\end{equation}
The boundary two-form has degree zero and
$Q_{\partial}\mathcal A_\partial=d_z\mathcal A_\partial$,
$Q_{\partial}\mathcal B_\partial=d_{z^{-1}}\mathcal B_\partial$.
It is Hamiltonian, with BFV charge represented by
$S_\partial=\int_{\partial X}\langle\mathcal B_\partial\wedge
d_z\mathcal A_\partial\rangle$.  The BV and BFV structures satisfy the
usual bulk--boundary compatibility relation
$\iota_{Q_X}\omega_X=\delta S_X+\pi^*\alpha_{\partial X}$,
up to the global convention for the sign of the boundary primitive.
For a detailed account and the modified quantum master equation, see
\cite{CMRClassical,CMRPerturbative}.

It is important that the boundary pairing in \eqref{eq:alphaBFV} is
between \emph{dual} local systems.  Replacing $L_z^\vee$ by $L_z$ without
a trivialization would in general make the integral meaningless.

\subsection{The sutured annuli as boundary conditions}
Choose two disjoint, oppositely oriented meridional sutures on
$\partial X_K\cong T^2$.  They cut the torus into two annuli; one is
$R_-$, with oriented meridional core, and the closure of its complement
is an annulus $R_-^c$.  For defining the relative complex we only require
the embedded annulus $R_-$; we do \emph{not} confuse $R_-^c$ with the
sutured $R_+$ after reserving annular neighborhoods of the sutures.
They differ by boundary collars and have the same homotopy type.

A mixed linear boundary condition is
\begin{equation}\label{eq:mixboundary}
 \mathcal A|_{R_-}=0,\qquad
 \mathcal B|_{R_-^c}=0.
\end{equation}
Consequently $\alpha_{\partial X}$ vanishes on variations satisfying
\eqref{eq:mixboundary}: at each point of the boundary, either
$\delta\mathcal A$ or $\mathcal B$ vanishes.  The boundary condition is
preserved by $Q_X$ since the exterior covariant differential commutes
with restriction.  In a cellular description with dual CW structures,
the two complementary complexes are naturally paired by
Poincar\'e--Lefschetz duality:
\begin{equation}\label{eq:PLduality}
 C^k(X_K,R_-;L_z)\otimes
 C^{3-k}(X_K,R_-^c;L_z^\vee)\longrightarrow\mathbb C.
\end{equation}
The precise chain-level pairing can be constructed with a dual cellular
decomposition.  The corresponding gauge-fixed quadratic BF model has
the underlying twisted differential
$C^\bullet(X_K,R_-;L_z)$.
Its one-loop determinant-line contribution is the torsion of this
\emph{relative} complex; this is the finite-dimensional manifestation
of the BV--BFV boundary-polarization principle.  It is not a claim that
the unregularized continuum path integral is an ordinary number.

The annuli meet along meridional circles.  Promoting
\eqref{eq:mixboundary} to a fully functorial quantum BV--BFV theory
with polarized interfaces requires the treatment of corners and coherent
state spaces.  We work at the \emph{cellular torsion level}, where the
relative complex and all of its gluing maps are defined without appealing
to a conjectural corner quantization.

\subsection{Acyclic boundary and corner sectors}
\begin{lemma}[No boundary or corner residual cohomology]\label{lem:corner}
Let $z\ne1$ and $R\subset\partial X_K$ be a meridional annulus.  The
rank-one twisted cellular complexes of $R$, of its two meridional
boundary circles, and of the entire torus $\partial X_K$ are acyclic.
If also $\Delta_K(z)\ne0$, then
$H^\bullet(X_K,R;L_z)=0$.
\end{lemma}
\begin{proof}
A meridional circle has the cellular cochain complex
\begin{equation}\label{eq:circle-complex}
 0\longrightarrow\mathbb C\xrightarrow{z-1}\mathbb C
 \longrightarrow0,
\end{equation}
which is exact.  An annulus retracts onto its meridional core.
For the torus, use a one-vertex, two-edge, one-face cellulation with
meridional holonomy $z$ and longitudinal holonomy $1$:
\begin{equation}\label{eq:torus-cpx}
 0\longrightarrow\mathbb C\xrightarrow{\,\binom{z-1}{0}\,}
 \mathbb C^2\xrightarrow{\,\begin{psmallmatrix}0&z-1\end{psmallmatrix}\,}
 \mathbb C\longrightarrow0.
\end{equation}
It is exact.  If $\Delta_K(z)\ne0$, the Fox complex of $X_K$ is also
exact.  Finally, the long exact sequence of the pair $(X_K,R)$
then proves that its relative twisted cohomology vanishes.
\end{proof}

The lemma means that the reduced boundary and corner \emph{classical}
cohomologies vanish in these sectors.  It does \emph{not} mean that all
unreduced boundary BFV fields disappear, or that one has constructed a
canonical equivalence with the nontrivial bordered Floer categories.
In fact, this acyclicity helps explain why the scalar theory cannot
retain Maslov-graded Floer information.

\subsection{Determinant lines, torsion and the normalization issue}
A self-contained algebraic definition is useful here.  For a finite
based cochain complex $C^\bullet$ of vector spaces over a field $E$,
put
\begin{equation}\label{eq:detline}
 \operatorname{Det}C^\bullet
 =\bigotimes_k(\det C^k)^{(-1)^k}.
\end{equation}
The decomposition into cycles, boundaries and cohomology induces a
canonical determinant-line isomorphism
$\operatorname{Det}C^\bullet\simeq
\operatorname{Det}H^\bullet(C)$.  If $C$ is acyclic and cellular lifts
specify bases, this map yields a nonzero scalar, its Reidemeister
torsion; changing from cochains to chains or reversing the convention
may invert this scalar.  We fix the \emph{Milnor chain convention}
from Proposition~\ref{prop:torsion} for all numerical formulas.  The
Gaussian/BV pushforward is being compared with the corresponding
chosen torsion branch, not given a new normalization-independent definition.

For a non-unimodular rank-one representation, replacing the chosen lift
of a $k$-cell by a deck translate changes the basis by a power of $z$.
The torsion therefore changes by a unit $\pm z^n$; this is precisely
the familiar $\pm t^n$ ambiguity of Alexander torsion.  The cellular
BV--BFV construction in \cite{CMR} is stated for coefficient systems
with determinant holonomy $\pm1$.  The specialization $L_\zeta$
typically does not meet this hypothesis.  In this paper we instead
make \emph{fixed-lift comparisons} of the finite-dimensional
quadratic complexes, followed by the stated symmetric Alexander
normalization.  This restriction is substantial: the full
non-unimodular quantum BV--BFV functor, with Euler-structure coherence,
is not constructed here.

\begin{proposition}[Boundary normalization as relative torsion]
\label{prop:relative}
For every knot $K$ and every $z\ne1$ such that
$\Delta_K(z)\ne0$, let $R_-$ be a meridional annulus of its sutured
exterior.  In the Milnor torsion convention one has
\begin{equation}\label{eq:relative-amplitude}
 \tau(X_K,R_-;L_z)\doteq\Delta_K(z),
 \qquad
 \frac{\tau(X_K;L_z)}{\tau(X_K,R_-;L_z)}
 \doteq\frac1{z-1}.
\end{equation}
Both equalities are statements about torsion of \emph{acyclic based
complexes}; a normalized value is obtained by symmetrizing the
Laurent polynomial.
\end{proposition}
\begin{proof}
Choose a CW decomposition of the pair $(X_K,R_-)$ with compatible
lifts and bases.  The restriction or quotient of cellular chains
provides a short exact sequence
\begin{equation}\label{eq:SESrelative}
 0\longrightarrow C_\bullet(R_-;L_z)
 \longrightarrow C_\bullet(X_K;L_z)
 \longrightarrow C_\bullet(X_K,R_-;L_z)
 \longrightarrow0.
\end{equation}
All three complexes are acyclic by Lemma~\ref{lem:corner}.
The multiplicativity property of torsion for a based short exact
sequence therefore yields
$\tau(X_K;L_z)\doteq
\tau(R_-;L_z)\tau(X_K,R_-;L_z)$.
The annulus has the homotopy type of \eqref{eq:circle-complex},
whose torsion in our convention is $(z-1)^{-1}$.
Proposition~\ref{prop:torsion} gives
$\tau(X_K;L_z)\doteq\Delta_K(z)/(z-1)$.
Division produces \eqref{eq:relative-amplitude}.
\end{proof}

\begin{definition}[Symmetrized relative BF torsion character]
\label{def:bfcharacter}
For a knot exterior, define
$\mathscr B_K(t)\in\mathbb Z[t,t^{-1}]$ to be the unique
representative of the meridional relative torsion class
$[\tau(X_K,R_-;F_\phi)]\in F^\times/(\pm t^{\mathbb Z})$
that is symmetric in $t\leftrightarrow t^{-1}$ and equals $1$ at $t=1$.
At any nonresonant character $z$, its value $\mathscr B_K(z)$ is called
the \emph{symmetrized relative BF torsion character}.
\end{definition}
\begin{remark}
Existence and uniqueness in Definition~\ref{def:bfcharacter} follow
from Proposition~\ref{prop:relative} and the normalization of
$\Delta_K$.  ``Character'' here denotes a function of the abelian
holonomy, not a character of a Floer-category representation.
The symmetric normalization is extra data \emph{added after}
cellular torsion: it is not a canonical value of a bare Gaussian integral.
\end{remark}

\begin{corollary}\label{cor:bfalex}
For every knot $K$ one has the identity of normalized Laurent
polynomials
\begin{equation}\label{eq:bf-alex}
 \mathscr B_K(t)=\Delta_K(t).
\end{equation}
For every nonresonant $z$, this is an equality of the symmetrically
normalized determinant-line contribution of the mixed cellular BF
boundary model with the Alexander polynomial evaluated at $z$.
\end{corollary}
\begin{proof}
This is Proposition~\ref{prop:relative} together with the uniqueness
of the symmetric Alexander normalization.
\end{proof}

\section{The Floer side: sutures and decategorification}
\label{sec:sutured}
\subsection{Sutured and knot Floer homology}
The knot Floer invariant $\widehat{HFK}_m(K,a)$ is a finitely
generated group (or vector space after passing to a field) with a
Maslov grading $m$ and an Alexander grading $a$.  It is computed
from a doubly pointed Heegaard diagram by counting suitable
holomorphic disks.  The groups themselves depend on substantially more
than an Euler-characteristic polynomial.  In our comparison all
Floer groups are taken over $\mathbb F_2$ and signs in their Euler
characteristics use the Maslov parity fixed by the usual absolute
normalization for knots in $S^3$.

The balanced sutured manifold $(X_K,\gamma_\mu)$ obtained using
\emph{two oppositely oriented meridional sutures} satisfies
$SFH(X_K,\gamma_\mu)\cong\widehat{HFK}(K)$ as an ungraded sum, with
the relevant Alexander grading recovered from the relative
$\operatorname{Spin}^c$ decomposition and a choice of Seifert surface
\cite{Juhasz}.  The theorem of Friedl--Juh\'asz--Rasmussen
\cite{FJR} identifies the relative torsion class of a balanced
sutured manifold with the Euler-characteristic class of its sutured
Floer homology, up to the conventional affine $H_1$ unit and sign.
The familiar Ozsv\'ath--Szab\'o identity then gives the normalized
formula
\begin{equation}\label{eq:suturedhfk}
 \sum_a\chi(\widehat{HFK}(K,a))t^a
 =\Delta_K(t)=\mathscr B_K(t).
\end{equation}
The left equality is the original knot Floer categorification
statement; the right equality is the explicit meridional relative
torsion calculation.  Their coexistence is our precisely defined
\emph{decategorified BV--BFV/Floer comparison}.

\begin{theorem}[Relative BV--BFV/Floer comparison and its exact scope]
\label{thm:BFHF}
Fix a knot $K$, a meridional sutured exterior, and the Milnor and
Alexander normalization conventions above.  The following triangle
commutes after passing to the common Laurent-polynomial decategorification:
\[
\begin{tikzcd}[column sep=large,row sep=large]
\mathsf{BF}_{\mathrm{rel}}(K)
 \arrow[r,"\mathrm{torsion}"]
 \arrow[dr,swap,"\mathrm{normalize}"]
 & \tau(X_K,R_-)
 \arrow[d,"\mathrm{normalize}"]\\
 & \displaystyle \Delta_K(t)=\sum_a\chi(\widehat{HFK}(K,a))t^a
\end{tikzcd}
\]
The two routes compute the same decategorified invariant.
No quasi-isomorphism, natural transformation, or pairing theorem
between the BV--BFV boundary-state complex and a bordered or sutured
Floer chain complex follows from this triangle.
\end{theorem}
\begin{proof}
The determinant-line assignment is the finite-dimensional
acyclic torsion construction of \eqref{eq:detline}.  Its relative
meridional evaluation is computed in
Proposition~\ref{prop:relative}, and its normalized representative
is $\mathscr B_K=\Delta_K$ by Corollary~\ref{cor:bfalex}.
The identification with the graded Euler-characteristic polynomial
is \eqref{eq:suturedhfk}, with the cited Floer theorems as input.
\end{proof}

\begin{remark}[The pre-existing categorification result]
Hom, Lidman and Watson proved a genuine categorification of
Donaldson's Alexander-module TQFT using bordered Floer theory
\cite{HLW}.  Their theorem provides substantially stronger Floer
data than our triangle; it does \emph{not}, as stated, identify that
bordered theory with a BV--BFV state.  Claiming such an identification
would require an independent construction and proof.
\end{remark}

\section{Euler coefficients and Alexander support}
Fix coefficients in $\F_2$ and keep the Maslov-parity convention from
Section~\ref{sec:sutured}.  For later arithmetic formulas, write the
coefficient vector of the decategorified Floer invariant as
\begin{equation}\label{eq:chi-coeff}
 c_a(K)=\chi\bigl(\HF(K,a)\bigr)
 =\sum_m(-1)^m\dim_{\F_2}\HF_m(K,a).
\end{equation}
Ozsv\'ath and Szab\'o's Euler characteristic identity, together with the genus bound, gives
\begin{equation}\label{eq:decategorification}
 \Delta_K(t)=\sum_{a\in\Z}c_a(K)t^a,
 \qquad c_a(K)=0\ \text{when }|a|>g(K).
\end{equation}
The symmetry $c_a=c_{-a}$ and the identity $\sum_a c_a=1$ follow from the corresponding Alexander-polynomial properties. Notice that $c_a$ does not determine $\HF(K,a)$ as a graded group: cancellations in Euler characteristic discard information.

The established torsion--Floer connection is considerably richer than \eqref{eq:decategorification}. Friedl--Juh\'asz--Rasmussen compare sutured Floer Euler characteristics with torsion; Hom--Lidman--Watson use bordered Floer theory to categorify Donaldson's Alexander-module TQFT \cite{FJR,HLW}. We do not construct a map between BV--BFV state complexes and bordered Floer modules here. Our contribution is a finite-character reconstruction statement for precisely the common \emph{decategorified} information.

\section{Exact cyclotomic reconstruction and optimality}
\subsection{A trace-orthogonality lemma}
Let $p$ be prime, $\zeta=e^{2\pi i/p}$, and $E_p=\Q(\zeta)$. The Galois group $\operatorname{Gal}(E_p/\Q)$ consists of $\sigma_j(\zeta)=\zeta^j$ for $1\le j\le p-1$.
\begin{lemma}[Prime cyclotomic trace]\label{lem:trace}
For any integer $m$,
\begin{equation}\label{eq:trace-orth}
 \Tr_{E_p/\Q}(\zeta^m)=\begin{cases}p-1,&p\mid m,\\-1,&p\nmid m.\end{cases}
\end{equation}
\end{lemma}
\begin{proof}
The trace is $\sum_{j=1}^{p-1}\zeta^{jm}$. If $p\mid m$, each summand is $1$. Otherwise multiplication by $m$ permutes the nonzero residue classes modulo $p$, so the sum is $\sum_{j=1}^{p-1}\zeta^j=-1$.
\end{proof}

\begin{proposition}[Exact extraction]\label{prop:extract}
Let $P(t)=\sum_{b=-g}^{g}c_bt^b\in\Z[t,t^{-1}]$ satisfy $P(1)=1$. If $p\ge2g+1$ is prime, then for every $-g\le a\le g$,
\begin{equation}\label{eq:extract}
 c_a=\frac{1+\Tr_{E_p/\Q}(\zeta^{-a}P(\zeta))}{p}.
\end{equation}
In particular the map $P\mapsto P(\zeta)$ is injective on the class of such normalized Laurent polynomials.
\end{proposition}
\begin{proof}
Since $|a-b|\le2g<p$, the only way $p$ can divide $b-a$ is $b=a$. Lemma \ref{lem:trace} therefore gives
\begin{align*}
 \Tr\bigl(\zeta^{-a}P(\zeta)\bigr)
 &=\sum_{b=-g}^{g}c_b\Tr(\zeta^{b-a})\\
 &=(p-1)c_a-\sum_{b\ne a}c_b
 =pc_a-P(1)=pc_a-1.
\end{align*}
Rearrange to obtain \eqref{eq:extract}. Equality of evaluation values implies equality of every $c_a$.
\end{proof}

\begin{proof}[Proof of Theorem \ref{thm:main}]
Lemma \ref{lem:boundary} and Lemma \ref{lem:nonvanish} give assertion (i). Proposition \ref{prop:torsion} gives $\Delta_K(\zeta)=(\zeta-1)\tau_K^{\mathrm{sym}}(\zeta)$, equivalently $\mathscr B_K(\zeta)=\Delta_K(\zeta)$ by Corollary~\ref{cor:bfalex}; apply Proposition \ref{prop:extract} to \eqref{eq:decategorification} to obtain \eqref{eq:main}. Assertion (iii) follows from injectivity in Proposition \ref{prop:extract}.
\end{proof}

\begin{corollary}[Arithmetic integrality constraints]\label{cor:integrality}
Under the hypotheses of Theorem \ref{thm:main}, let $\alpha=(\zeta-1)\tau_K^{\mathrm{sym}}(\zeta)=\Delta_K(\zeta)$. Then
\begin{equation}\label{eq:congruences}
 \Tr_{E_p/\Q}(\zeta^{-a}\alpha)\equiv -1\pmod p
 \qquad (-g\le a\le g),
\end{equation}
and these traces determine the entire Euler-characteristic polynomial. In particular, $\alpha$ is an algebraic integer, fixed by complex conjugation.
\end{corollary}
\begin{proof}
The congruence is exactly $\Tr(\zeta^{-a}\alpha)=pc_a(K)-1$. Integrality is clear from $\alpha\in\Z[\zeta]$; symmetry of $\Delta_K$ gives $\overline\alpha=\Delta_K(\zeta^{-1})=\alpha$.
\end{proof}

\begin{remark}[What ``one value'' means]\label{rem:exact}
A single exact element of $E_p$ is not the same as an unstructured complex floating-point approximation to one embedding of that element. Formula \eqref{eq:main} requires the exact field trace, equivalently all Galois conjugates determined by an exact representation of $\alpha$. Small numerical error in a single complex approximation is not by itself sufficient to recover an exact algebraic number without further height or coefficient bounds.
\end{remark}

\subsection{A sharp counterexample below the genus threshold}
We now prove that the inequality $p\ge2g+1$ is the precise uniform threshold when one allows all knots of genus at most $g$. Define
\begin{equation}\label{eq:obstructor}
 R_p(t)=t^{-(p-1)/2}\Phi_p(t)\bigl(t+t^{-1}-2\bigr),
 \qquad P_p(t)=1+R_p(t).
\end{equation}
\begin{lemma}\label{lem:obstructor}
For any odd prime $p$, the Laurent polynomial $P_p$ satisfies
\[
 P_p(t^{-1})=P_p(t),\qquad P_p(1)=1,\qquad P_p(\zeta)=1,
\]
and is monic, with highest and lowest exponents equal to $\pm(p+1)/2$.
\end{lemma}
\begin{proof}
The normalized cyclotomic polynomial $t^{-(p-1)/2}\Phi_p(t)$ is invariant under $t\leftrightarrow t^{-1}$; so is $t+t^{-1}-2$. The latter factor vanishes at $t=1$, while $\Phi_p(\zeta)=0$. The highest exponent in $R_p$ is $(p-1)/2+1=(p+1)/2$ and its coefficient is $1$. The same holds at the lowest exponent by symmetry, and adding $1$ does not change these extreme terms.
\end{proof}

\begin{proof}[Proof of Theorem \ref{thm:sharp}]
By Lemma \ref{lem:obstructor}, $P_p(t)$ is an admissible monic knot Alexander polynomial. A realization theorem of Stoimenow \cite{Stoimenow} produces a fibered knot $J_p$ whose normalized Alexander polynomial is $P_p$. For fibered knots, the genus equals the maximal exponent in the normalized Alexander polynomial; therefore $g(J_p)=(p+1)/2$. Since $P_p(\zeta)=1=\Delta_U(\zeta)$, division by $\zeta-1$ yields equality of their symmetrically normalized torsions at $\zeta$, although $P_p\ne1$. This knot lies in the genus-at-most-$g$ class when $g\ge(p+1)/2$, equivalently $p<2g+1$ for odd $p$.
\end{proof}

\begin{remark}\label{rem:small-primes}
The obstruction does not come from a zero of $\Delta_K(\zeta)$: prime cyclotomic Alexander evaluations never vanish. Rather, $\Phi_p(t)$ lies in the kernel of evaluation at $\zeta$, and multiplying it by the symmetric factor $t+t^{-1}-2$ preserves $P(1)=1$. The genus bound rules out exactly this type of aliasing when $p\ge2g+1$.
\end{remark}

\section{Cyclotomic norms and branched covers}
Let $\Sigma_p(K)$ be the $p$-fold cyclic branched cover of $S^3$ along $K$. The classical Fox resultant formula states that, whenever the resulting group is finite,
\begin{equation}\label{eq:foxbranched}
 |H_1(\Sigma_p(K);\Z)|
 =\left|\prod_{j=1}^{p-1}\Delta_K(\zeta^j)\right|.
\end{equation}
For prime $p$, Lemma \ref{lem:nonvanish} guarantees finiteness. The product on the right is an algebraic field norm.
\begin{theorem}[Simultaneous trace and norm recovery]\label{thm:norm}
Suppose $p\ge2g(K)+1$ is an odd prime. Set $\alpha_K=(\zeta-1)\tau_K^{\mathrm{sym}}(\zeta)$. Then the single exact element $\alpha_K\in\Z[\zeta]$ determines both
\begin{equation}\label{eq:norm}
 |H_1(\Sigma_p(K);\Z)|
 =\left|\Nm_{E_p/\Q}(\alpha_K)\right|
\end{equation}
and, for every Alexander grading $a$,
\begin{equation}\label{eq:trace-again}
 \chi\bigl(\HF(K,a)\bigr)=\frac{1+\Tr_{E_p/\Q}(\zeta^{-a}\alpha_K)}{p}.
\end{equation}
For odd $p$, the norm in \eqref{eq:norm} is positive.
\end{theorem}
\begin{proof}
The trace identity is Theorem \ref{thm:main}. The norm identity follows from \eqref{eq:foxbranched} and the fact that $\Delta_K(\zeta^j)=\sigma_j(\alpha_K)$. Since the embeddings indexed by $j$ and $p-j$ take the same real nonzero value on the symmetric Alexander polynomial, the norm is a product of positive squares.
\end{proof}
\begin{remark}
The trace and norm yield different information from the same arithmetic datum. The norm alone is far from injective on Alexander polynomials, and neither the trace nor the norm determines the full bigraded knot Floer group. Theorem \ref{thm:norm} is a package of two known types of topological information, not a recovery of Floer differentials or torsion-sensitive chain complexes.
\end{remark}

\section{Connected-sum gluing and Floer convolution}\label{sec:gluing}
Let $K_1$ and $K_2$ be knots. Their connected-sum exterior decomposes along a properly embedded meridional annulus $A\simeq S^1\times[0,1]$ into the two exteriors, with the same meridional holonomy $t$ on all three pieces. The corresponding local system on $A$ is acyclic when $t\ne1$, and its torsion in our convention is $\tau(A;t)=(t-1)^{-1}$.

\subsection{Relative gluing is the natural sutured multiplication}
The factor $t-1$ appearing in absolute gluing vanishes when we
formulate the invariant in the meridional relative boundary
condition.  This is not an ad hoc renormalization: the annulus is
already part of the sutured boundary data.

\begin{proposition}[Relative multiplicativity]\label{prop:relative-glue}
For knots $K_1,K_2$ and $K=K_1\#K_2$, choose compatible meridional
relative boundary data.  Their symmetrized relative BF characters
satisfy the \emph{exact} identity
\begin{equation}\label{eq:relglue}
 \mathscr B_{K_1\#K_2}(t)
 =\mathscr B_{K_1}(t)\mathscr B_{K_2}(t).
\end{equation}
Before symmetrization the identity is an equality of torsion
classes modulo $\pm t^{\mathbb Z}$, compatible with the gluing
of determinant lines in the cellular BF system.
\end{proposition}
\begin{proof}
Decompose the exterior along the meridional connected-sum annulus
$A$.  For nontrivial holonomy the Mayer--Vietoris cellular sequence
\[
0\to C_\bullet(A;L_z)\to
C_\bullet(X_{K_1};L_z)\oplus C_\bullet(X_{K_2};L_z)
\to C_\bullet(X_K;L_z)\to0
\]
is a short exact sequence of acyclic complexes.  Its determinant
lines multiply with an inverse factor $\tau(A;z)^{-1}=z-1$,
which gives Proposition~\ref{prop:gluing} below.  Dividing each
absolute torsion by the torsion of its meridional boundary annulus
now gives
\[
\frac{\tau(X_K;z)}{\tau(R_-;z)}
\doteq
\frac{\tau(X_{K_1};z)}{\tau(R_{-,1};z)}
\frac{\tau(X_{K_2};z)}{\tau(R_{-,2};z)}.
\]
The same identity follows from the relative Mayer--Vietoris
sequences after replacing the boundary annuli by compatible
subcomplexes.  Finally, symmetric representatives have constant
value $1$ at $t=1$, so the residual sign and monomial ambiguity is
removed.  This proves \eqref{eq:relglue}.
\end{proof}

\begin{proposition}[Scalar torsion gluing]\label{prop:gluing}
With the symmetric normalization \eqref{eq:tau-def},
\begin{equation}\label{eq:glue}
 \tau_{K_1\#K_2}^{\mathrm{sym}}(t)
 =(t-1)\tau_{K_1}^{\mathrm{sym}}(t)\tau_{K_2}^{\mathrm{sym}}(t).
\end{equation}
This agrees with the multiplicativity of acyclic torsion for the Mayer--Vietoris short exact sequence, after correcting by the annulus torsion.
\end{proposition}
\begin{proof}
The standard connected-sum identity is $\Delta_{K_1\#K_2}=\Delta_{K_1}\Delta_{K_2}$. Dividing this identity by $t-1$ gives \eqref{eq:glue} exactly, with all symmetric normalizations fixed. Independently, cellular Mayer--Vietoris gives a short exact sequence of acyclic twisted complexes $0\to C_*(A)\to C_*(X_{K_1})\oplus C_*(X_{K_2})\to C_*(X_{K_1\#K_2})\to0$. Multiplicativity of torsion gives $\tau(X_{K_1\#K_2})\doteq\tau(X_{K_1})\tau(X_{K_2})/\tau(A)$, whence the factor $t-1$; symmetric normalization removes the units.
\end{proof}

\begin{theorem}[Trace--gluing compatibility]\label{thm:convolution}
Let $g_i\ge g(K_i)$ and take an odd prime $p\ge2(g_1+g_2)+1$. For $\zeta=e^{2\pi i/p}$, set $\alpha_i=(\zeta-1)\tau_{K_i}^{\mathrm{sym}}(\zeta)=\Delta_{K_i}(\zeta)$. Then
\begin{equation}\label{eq:glue-trace}
 \frac{1+\Tr_{E_p/\Q}(\zeta^{-a}\alpha_1\alpha_2)}p
 =\sum_{b\in\Z}c_b(K_1)c_{a-b}(K_2)
 =c_a(K_1\#K_2).
\end{equation}
Equivalently, the torsion-gluing product corresponds under cyclotomic coefficient extraction to the convolution of knot Floer Euler characteristics.
\end{theorem}
\begin{proof}
Multiplicativity gives $\alpha_1\alpha_2=\Delta_{K_1\#K_2}(\zeta)$; the degree support is contained in $[-(g_1+g_2),g_1+g_2]$. Apply Proposition \ref{prop:extract} to the product. The coefficient of $t^a$ in $\Delta_{K_1}(t)\Delta_{K_2}(t)$ is the displayed convolution; it equals the Euler characteristic of the knot Floer connected-sum tensor product by the connected-sum formula of \cite{OzsvathSzaboKnots}.
\end{proof}

\begin{remark}[Scope of gluing]
Proposition \ref{prop:gluing} is a scalar torsion statement, which is compatible with the expected determinant-line behavior of Abelian BF theory. It is \emph{not} a construction of a BV--BFV gluing equivalence with the bordered Floer pairing theorem. Such a categorified comparison would require additional chain-level and boundary-polarization data, and is not proved here.
\end{remark}

\section{Examples}
\subsection{The trefoil and the figure-eight knot}
Both the right-handed trefoil $3_1$ and the figure-eight knot $4_1$ have genus $1$, so $p=3$ is admissible. Let $\zeta=e^{2\pi i/3}$, so $\zeta+\zeta^{-1}=-1$ and $E_3=\Q(\sqrt{-3})$.

For the trefoil,
\begin{equation}\label{eq:trefoil}
 \Delta_{3_1}(t)=t^{-1}-1+t,\qquad
 \alpha_{3_1}=\Delta_{3_1}(\zeta)=-2.
\end{equation}
The trace of the rational element $-2$ is $-4$, and
\[
 \Tr(\zeta^{-1}\alpha_{3_1})=-2(\zeta+\zeta^{-1})=2.
\]
Consequently $c_0=(1-4)/3=-1$ and $c_{\pm1}=(1+2)/3=1$. The norm is $\Nm(-2)=4$, recovering $|H_1(\Sigma_3(3_1))|=4$.

For the figure-eight,
\begin{equation}\label{eq:figure-eight}
 \Delta_{4_1}(t)=3-t-t^{-1},\qquad \alpha_{4_1}=4.
\end{equation}
Now $\Tr(4)=8$ and $\Tr(4\zeta^{-1})=-4$, so $c_0=3$ and $c_{\pm1}=-1$. The norm gives $|H_1(\Sigma_3(4_1))|=16$.

\subsection{The cinquefoil and the sharpness threshold}
The $(2,5)$ torus knot has genus $2$ and
\begin{equation}\label{eq:cinquefoil}
 \Delta_{T(2,5)}(t)=t^{-2}-t^{-1}+1-t+t^2.
\end{equation}
At a primitive cube root $\zeta_3$, one has $t^{-2}+t^2=t+t^{-1}$, hence $\Delta_{T(2,5)}(\zeta_3)=1=\Delta_U(\zeta_3)$. Therefore a single order-three torsion sector cannot distinguish even the Euler-characteristic polynomials of these two knots. The obstruction polynomial \eqref{eq:obstructor} for $p=3$ is precisely $\Delta_{T(2,5)}$.

At a primitive fifth root $\zeta_5$, by contrast, Theorem \ref{thm:main} applies. Set $s=\zeta_5+\zeta_5^{-1}$. Since $1+s+(\zeta_5^2+\zeta_5^{-2})=0$,
\[
 \alpha_{T(2,5)}=1-s+(\zeta_5^2+\zeta_5^{-2})=-2s.
\]
The trace identity gives
\[
 \Tr(\alpha)=4,\qquad \Tr(\zeta_5^{-1}\alpha)=-6,\qquad
 \Tr(\zeta_5^{-2}\alpha)=4,
\]
therefore $c_0=1$, $c_{\pm1}=-1$, and $c_{\pm2}=1$, as expected. The field norm is $\Nm_{E_5/\Q}(-2s)=16$.

\subsection{Connected sum of two trefoils}
For $K=3_1\#3_1$, the genus is $2$ and
\[
 \Delta_K(t)=(t^{-1}-1+t)^2
 =t^{-2}-2t^{-1}+3-2t+t^2.
\]
The first admissible prime is $p=5$. In particular the Alexander-graded Euler characteristics are $(1,-2,3,-2,1)$ in degrees $-2,-1,0,1,2$, agreeing with the convolution of the trefoil coefficients. Torsion gluing gives $\tau_{K}^{\mathrm{sym}}=(t-1)(\tau_{3_1}^{\mathrm{sym}})^2$, and Theorem \ref{thm:convolution} reconstructs exactly those coefficients from the fifth-root torsion value.

\section{Limits of reconstruction from scalar torsion}\label{sec:nogo}
A natural temptation is to regard the collection of all complex rank-one BF torsion sectors as a model of knot Floer homology itself, rather than of its Euler characteristic. The following elementary obstruction specifies why a further enhancement is necessary.

\begin{proposition}[Scalar torsion does not determine knot Floer homology]\label{prop:nogo}
There is no assignment that recovers the isomorphism type of the Alexander-graded groups $\HF(K,a)$ from the family of scalar normalized rank-one torsions $\{\tau_K^{\mathrm{sym}}(z)\}_{z\ne1}$ alone, even when all values, or equivalently the entire rational torsion function, are given.
\end{proposition}
\begin{proof}
Let $U$ be the unknot and let $KT$ be a nontrivial Kinoshita--Terasaka knot of genus $2$. The latter satisfies $\Delta_{KT}(t)=1=\Delta_U(t)$, a standard example discussed in the knot Floer calculations of Ozsv\'ath and Szab\'o \cite{OzsvathSzaboMutation}. Hence both knots have identical normalized rank-one torsions $1/(t-1)$, and therefore identical evaluations at every $z\ne1$. Nevertheless $\HF(U,a)=0$ for all $a\ne0$, whereas genus detection for knot Floer homology yields $\HF(KT,2)\ne0$ \cite{OzsvathSzaboGenus}. Their Alexander-graded knot Floer homologies cannot be isomorphic.
\end{proof}

\begin{remark}
The proposition is a limitation on reconstruction from \emph{scalar Abelian torsion}. It does not rule out categorification using additional BV--BFV boundary structures, non-Abelian sectors, observables, or a Floer-theoretic chain complex. It also explains why even an exact knowledge of all the cyclotomic trace and norm data in the preceding sections stops at decategorification. 
\end{remark}

\section{Linear BV--BFV gluing}
\label{sec:chaininterface}
The preceding comparison uses torsion, which discards the differential
up to acyclic simple-homotopy information.  We now retain the entire
\emph{cellular cochain complexes} long enough to implement cutting and
regluing by an explicit chain homotopy.  The construction concerns the linear BV--BFV model and its cellular
restriction maps.

\subsection{The boundary \texorpdfstring{$Q$}{Q}-complex and two explicit contractions}
Let $k$ be a field and $z\in k^\times$ with $z\neq1$, and put
$u=z-1$.  Fix a cellular decomposition of a boundary interface
$Y$ together with its twisted local system $L_z$.
The linearized $\mathcal A$-part of the BFV field complex is
$C^\bullet(Y;L_z)[1]$; the $\mathcal B$-part is the dual-coefficient
complex $C^\bullet(Y;L_z^\vee)[1]$ on the dual cell structure.
Their $Q$-operators are the cellular coboundaries.  The finite-dimensional
boundary symplectic pairing is induced by evaluation
$L_z\otimes L_z^\vee\to k$ and cellular Poincar\'e duality.  These
statements specify the \emph{linearized boundary fields}; a quantum
BFV state complex also requires a polarization, a measure and
possibly corner data.

\begin{lemma}[Explicit meridional interface contractions]
\label{lem:explicit-contract}
For $z\neq1$, the following two complexes admit the displayed
contracting homotopies $h$ satisfying $dh+hd=1$ and $h^2=0$:
\begin{enumerate}[label=\textup{(\roman*)},leftmargin=2em]
\item On a meridional annulus $A$, deformation retracting onto its
core circle, use
\[
D_A^0=k\xrightarrow{\,u\,}D_A^1=k,
\qquad h_A:D_A^1\to D_A^0,\quad h_A(b)=u^{-1}b.
\]
\item On a torus $T^2$ with meridional holonomy $z$ and longitudinal
holonomy $1$, use the complex
\[
D_T^0=k\xrightarrow{\,d^0\,}D_T^1=k^2
\xrightarrow{\,d^1\,}D_T^2=k,
\quad d^0(c)=(uc,0),\quad d^1(a,b)=ub,
\]
with
\[
h_T^1(a,b)=u^{-1}a,\qquad h_T^2(c)=(0,u^{-1}c),\qquad h_T^0=0.
\]
\end{enumerate}
The corresponding complexes with coefficients in $L_z^\vee$
contract by the same formulas with $u$ replaced by $z^{-1}-1$.
\end{lemma}
\begin{proof}
For the annulus, $dh$ and $hd$ act as the identity in degrees $1$
and $0$, respectively.  For the torus, one computes
\[
 h_T^1d^0(c)=c,\qquad
 (d^0h_T^1+h_T^2d^1)(a,b)=(a,0)+(0,b)=(a,b),
 \qquad d^1h_T^2(c)=c.
\]
All compositions $h^2$ vanish by the displayed definitions.
Dualizing the character simply replaces its meridional holonomy.
\end{proof}

\begin{remark}[The resonant point]
\label{rem:resonance}
The contraction contains the coefficient $(z-1)^{-1}$ and is
therefore meromorphic in the holonomy.  It has no extension to a
contracting homotopy over $k[z,z^{-1}]$ near $z=1$: already for the
meridional circle, $u h=1$ would require $u$ to be a unit.
At $z=1$ the cohomology of both the annulus and torus is nonzero.
Thus the nonresonant chain-level gluing mechanism below cannot be
continued unchanged through trivial holonomy.  This pole is a
property of the \emph{explicit linear complex}, not an assertion
about a singularity of any completed quantum BV--BFV theory.
\end{remark}

\subsection{A homotopy pullback with an explicit strong deformation retract}
Consider a finite CW decomposition of a pair $(X,R)$ obtained by
gluing two CW pairs along a common subpair:
\begin{equation}\label{eq:gluing-pairs}
 (X,R)=(X_1,R_1)\cup_{(Y,S)}(X_2,R_2),\quad
 Y=X_1\cap X_2,\quad S=R\cap Y,\quad R_i=R\cap X_i.
\end{equation}
Choose a flat rank-one local system $L$ on $X$; all restrictions
are understood.  Write
\[
 C_i=C^\bullet(X_i,R_i;L),\qquad
 D=C^\bullet(Y,S;L),\qquad
 r=r_1-r_2:C_1\oplus C_2\longrightarrow D.
\]
Let $P=\operatorname{holim}(C_1\xrightarrow{r_1}D\xleftarrow{r_2}C_2)$
be represented by the shifted mapping cone:
\begin{align}
 P^n&=C_1^n\oplus C_2^n\oplus D^{n-1},\label{eq:pullback-object}\\
 d_P(x_1,x_2,e)&=(d_1x_1,d_2x_2,
                 r_1x_1-r_2x_2-d_De).
 \label{eq:pullback-differential}
\end{align}
The sign $-d_D$ in the third component makes $d_P^2=0$.
The strict cellular matching complex $C^\bullet(X,R;L)$ maps to
$P$ by restriction $x\mapsto(x|_{X_1},x|_{X_2},0)$.

\begin{theorem}[Constructive acyclic-interface gluing]
\label{thm:chain-gluing}
Suppose that the cellular interface complex $D$ admits a
contraction $h:D^n\to D^{n-1}$ with
$dh+hd=1$ and $h^2=0$.  Then:
\begin{enumerate}[label=\textup{(\roman*)},leftmargin=2em]
\item The cellular restriction
$C^\bullet(X,R;L)\to P$ is a quasi-isomorphism.
\item The maps
\begin{align*}
 q:P&\longrightarrow C_1\oplus C_2,
 &q(x_1,x_2,e)&=(x_1,x_2),\\
 s:C_1\oplus C_2&\longrightarrow P,
 &s(x_1,x_2)&=(x_1,x_2,h(r_1x_1-r_2x_2)),\\
 H:P^n&\longrightarrow P^{n-1},
 &H(x_1,x_2,e)&=(0,0,-he)
\end{align*}
are chain maps in the case of $q,s$ and a homotopy in the case of
$H$. They give a strong deformation retract:
\begin{equation}\label{eq:SDR}
 qs=1,\qquad d_PH+Hd_P=1-sq,
 \qquad Hs=qH=H^2=0.
\end{equation}
In particular, the \emph{glued cochain complex} is related by an
explicit zigzag of quasi-isomorphisms to $C_1\oplus C_2$.
\item If all three complexes $C_1,C_2,D$ are acyclic and their
cellular bases and lifts are compatible, their Milnor torsions obey
\begin{equation}\label{eq:chain-glue-torsion}
 \tau(X,R;L)\doteq
 \frac{\tau(X_1,R_1;L)\tau(X_2,R_2;L)}{\tau(Y,S;L)}.
\end{equation}
The equality is in the appropriate determinant line before
choosing compatible bases; the displayed scalar formula retains
the usual deck-translation and sign ambiguity.
\end{enumerate}
This theorem concerns the \emph{linearized cellular BV--BFV
$Q$-complex} and its torsion, not a quasi-isomorphism of interacting
quantum boundary state complexes.
\end{theorem}
\begin{proof}
The cochain-level cellular Mayer--Vietoris sequence for
\eqref{eq:gluing-pairs} is degreewise exact:
\begin{equation}\label{eq:MV-cochain-short}
0\longrightarrow C^\bullet(X,R;L)
\longrightarrow C_1\oplus C_2\xrightarrow{\,r\,}D
\longrightarrow0.
\end{equation}
The map $r$ is degreewise surjective: any cellular cochain on a
cell of $Y\setminus S$ extends by zero to a cochain on $X_1$.
The associated homotopy-pullback comparison of the strict kernel
with $\operatorname{Cone}(r)[-1]$ is a quasi-isomorphism.  More
explicitly, the quotient of the cone by the strict kernel is an
acyclic cone of the surjective map after choosing a linear splitting
of $r$ in each degree.  This proves (i).

For (ii), abbreviate $x=(x_1,x_2)$ and $d_C=d_1\oplus d_2$.
Since $r$ is a cochain map,
\[
 d_Ps(x)=(d_Cx,rx-d_Dhrx)
       =(d_Cx,hd_Drx)=s(d_Cx).
\]
Thus $s$ is a chain map, as is the evident projection $q$.
Furthermore,
\[
(d_PH+Hd_P)(x,e)
 =(0,d_Dhe+hd_De-hrx)=(0,e-hrx)
 =(1-sq)(x,e).
\]
The side conditions follow from $h^2=0$ and the formulas.
This also exhibits a concrete chain homotopy inverse to the
projection $P\to C_1\oplus C_2$.

For (iii), use the \emph{chain} version of
\eqref{eq:MV-cochain-short},
$0\to C_\bullet(Y,S;L)\to C_\bullet(X_1,R_1;L)
\oplus C_\bullet(X_2,R_2;L)\to C_\bullet(X,R;L)\to0$.
The multiplicativity theorem for based acyclic complexes gives
$\tau(C_1)\tau(C_2)\doteq\tau(D)\tau(C_X)$.
Compatibility of the cell lifts identifies these with the displayed
torsions, proving \eqref{eq:chain-glue-torsion}.
\end{proof}

\begin{corollary}[Annular and toroidal gluing]
\label{cor:annular-toroidal}
For a meridional interface $Y=A\simeq S^1\times[0,1]$ with
$S=\varnothing$, the relative BF cochain complex admits the
explicit retract \eqref{eq:SDR} for $z\ne1$, and
\begin{equation}\label{eq:annular-torsion-chain}
\tau(X,R;L_z)\doteq(z-1)\,
\tau(X_1,R_1;L_z)\tau(X_2,R_2;L_z)
\end{equation}
whenever the complexes are acyclic.  For a torus interface
$Y=T^2$ with meridional holonomy $z$ and longitudinal holonomy
$1$, the same chain-level retract exists and
\begin{equation}\label{eq:torus-torsion-chain}
\tau(X,R;L_z)\doteq
\tau(X_1,R_1;L_z)\tau(X_2,R_2;L_z),
\end{equation}
since $\tau(T^2;L_z)\doteq1$.
The statements hold for the $L_z^\vee$ part of the linear BV--BFV
fields as well.
\end{corollary}
\begin{proof}
Apply Lemma~\ref{lem:explicit-contract} to
Theorem~\ref{thm:chain-gluing}.  The Milnor torsion of a
meridional circle (hence annulus) is $(z-1)^{-1}$, while the torus
complex splits as the sum of a two-term complex with boundary $u$
and a shifted two-term complex with the same boundary.  Their
torsions are inverse, giving torsion $1$ up to sign.
\end{proof}

\begin{remark}[What the deformation retract does \emph{not} do]
The maps $s,q,H$ depend on a cellular decomposition, the local
system and the contracting homotopy.  In particular, $h$ uses
$(z-1)^{-1}$.  We do not claim these maps implement the quantum
BFV integral kernel, respect a chosen boundary polarization
symplectically, or are compatible with a bordered Floer $A_\infty$
module structure.  Such claims require a separate quantization and
comparison construction.  What is established is a chain-level
\emph{homological gluing} theorem for the linearized BFV complexes,
not merely an equality of their scalar torsions.
\end{remark}

\section{Floer gluing and decategorification}
\label{sec:twotrack}
\subsection{Two different gluing operations}
The natural gluing operations of the two theories are different.
The linear cellular BFV interface is represented by a homotopy
pullback of cochain complexes; in bordered Floer theory the
boundary carries a differential graded algebra $\mathcal A(F)$,
and a bordered three-manifold has a type $D$ structure or an
$A_\infty$ type $A$ module.  The pairing theorem of
Lipshitz--Ozsv\'ath--Thurston identifies the chain complex of a
\emph{glued closed manifold} with a suitable box tensor product
\cite{LOT}.  The categorification theorem of Hom--Lidman--Watson
identifies Donaldson's Alexander-module TQFT as a decategorification
of a bordered Floer construction \cite{HLW}.

For orientation, if $\mathcal A$ is a differential graded algebra over
$\F_2$, a type $D$ structure is a graded vector space $N$
with a map $\delta^1:N\to\mathcal A\otimes N$ satisfying
\[
 (d_{\mathcal A}\otimes 1)\delta^1+
 (\mu\otimes1)(1\otimes\delta^1)\delta^1=0.
\]
An $A_\infty$ right module $M$ has structure maps
$m_{j+1}:M\otimes\mathcal A^{\otimes j}\to M$
obeying the $A_\infty$ relations.  Writing
$\delta^j:N\to\mathcal A^{\otimes j}\otimes N$
for the iterates of $\delta^1$ and $\delta^0=1$, the box
tensor differential on $M\otimes N$ is schematically
\[
 \partial_{\boxtimes}(m\otimes n)
 =\sum_{j\ge0}(m_{j+1}\otimes1)
   (m\otimes\delta^j(n)).
\]
Under the usual boundedness hypothesis the sum is finite, and
the defining relations imply $\partial_{\boxtimes}^2=0$.
This operation differs from the
homotopy pullback \eqref{eq:pullback-object}: the latter uses only
linear restriction maps of cellular cochains, rather than the
higher operations of a bordered algebra.

These theorems do not give a functor from the cellular BFV boundary
$Q$-complex to the bordered algebra $\mathcal A(F)$, nor a
natural quasi-isomorphism between the associated state spaces.
For knot connected sums, however, the knot Floer K\"unneth formula
is an especially transparent \emph{chain-level} pairing in its own
right, and it can be compared exactly with the BF torsion of an
annular gluing.

For a knot $K$, let
$\mathcal F_K=\operatorname{gr}_A\widehat{CFK}(K;\F_2)$
be the associated graded complex of the Alexander-filtered hat
knot Floer chain complex.  It is a finite bigraded chain complex;
its differential preserves Alexander grading and decreases Maslov
grading by one.  Write
\[
 \mathrm{Eu}(\mathcal F_K;t)=
 \sum_{a\in\Z}\left(\sum_m(-1)^m
 \dim_{\F_2} H_m(\mathcal F_K)_a\right)t^a.
\]
The Euler characteristic can equivalently be calculated from the
chain groups, as long as the differential preserves Alexander degree.

\begin{theorem}[Compatibility of torsion and Euler characteristic under gluing]
\label{thm:two-track}
Let $K_1,K_2$ be oriented knots and let
$K=K_1\#K_2$ be their connected sum.  Work with a nontrivial
meridional character $z$ for which all relevant twisted complexes
are acyclic, with the normalization and lift conventions of
Sections~\ref{sec:BVBFV}--\ref{sec:sutured}.
\begin{enumerate}[label=\textup{(\roman*)},leftmargin=2em]
\item On the BF side, the cellular homotopy pullback for the
meridional connected-sum annulus admits the explicit strong
deformation retract of Theorem~\ref{thm:chain-gluing}.
Its torsion is the annular gluing product
\[
 \tau(X_K;L_z)\doteq(z-1)
 \tau(X_{K_1};L_z)\tau(X_{K_2};L_z).
\]
\item On the Floer side, the knot Floer connected-sum formula gives
an Alexander- and Maslov-compatible quasi-isomorphism (over $\F_2$)
\begin{equation}\label{eq:Floer-Kunneth-chain}
 \mathcal F_{K_1\#K_2}
 \simeq\mathcal F_{K_1}\otimes_{\F_2}\mathcal F_{K_2},
\end{equation}
where the tensor differential is the usual sum and the Alexander
and Maslov degrees are added.
\item After applying the \emph{respective} decategorification maps,
namely symmetric relative torsion for the BF side and graded Euler
characteristic for the Floer side, these chain-level gluing
operations give the same identity
\begin{equation}\label{eq:twotrack-output}
\quad
 \mathscr B_{K_1\#K_2}(t)
 =\mathscr B_{K_1}(t)\mathscr B_{K_2}(t)
 =\mathrm{Eu}(\mathcal F_{K_1};t)
  \mathrm{Eu}(\mathcal F_{K_2};t).
 \quad
\end{equation}
\end{enumerate}
The result is a \emph{pairing-compatible comparison of the
Euler/torsion shadows of two chain-level gluing procedures}; it
is not a natural transformation between the two procedures.
\end{theorem}
\begin{proof}
Part (i) is Corollary~\ref{cor:annular-toroidal} applied to the
connected-sum decomposition along a meridional annulus.  For
(ii), use the connected-sum K\"unneth theorem for the knot Floer
complex \cite{OzsvathSzaboKnots}.  The tensor differential
preserves the sum of Alexander gradings, so its Euler
characteristic is multiplicative:
\[
 \mathrm{Eu}(\mathcal F_{K_1}\otimes\mathcal F_{K_2};t)
 =\mathrm{Eu}(\mathcal F_{K_1};t)
  \mathrm{Eu}(\mathcal F_{K_2};t).
\]
For (iii), the relative torsion comparison gives
$\mathscr B_{K_i}(t)=\Delta_{K_i}(t)$ for each knot and the
connected sum.  The Ozsv\'ath--Szab\'o graded Euler characteristic
formula identifies
$\mathrm{Eu}(\mathcal F_{K_i};t)=\Delta_{K_i}(t)$.
Together with the multiplicativity proved on both sides, these
give \eqref{eq:twotrack-output}.
\end{proof}

\subsection{An obstruction to an ordinary chain quasi-isomorphism}
There is a simple obstruction to strengthening Theorem~\ref{thm:two-track}
by identifying the standard complexes themselves.

\begin{theorem}[Acyclicity obstruction to a direct BF--Floer quasi-isomorphism]
\label{thm:chain-nogo}
Let $k=\F_2(t)$, and endow a knot exterior and its meridional
relative pair with the rank-one local system defined by the
abelianization $\pi_1(X_K)\to\Z$ followed by $1\mapsto t$.
For every knot $K\subset S^3$,
\begin{equation}\label{eq:chain-nogo}
 H^\bullet(X_K,R_-;k_\phi)=0,
 \qquad
 H_\bullet(\mathcal F_K\otimes_{\F_2}k)
 =\widehat{HFK}(K;\F_2)\otimes_{\F_2}k\ne0.
\end{equation}
Consequently there is no quasi-isomorphism in the ordinary
derived category $D(k)$ between the standard twisted relative
cellular BF complex and the scalar extension of the associated
graded knot Floer complex.  The obstruction remains if one allows
an overall cohomological shift.  Any genuinely categorified
BV--BFV/Floer comparison must change its objects or introduce
additional non-acyclic boundary/categorical data.
\end{theorem}
\begin{proof}
The Alexander polynomial satisfies $\Delta_K(1)=1$, so its
reduction modulo $2$ is nonzero in $\F_2[t,t^{-1}]$.  The
Fox complex over the fraction field $\F_2(t)$ is acyclic by
invertibility of an Alexander--Fox minor and of $t-1$.
The meridional annulus is also acyclic because $t-1$ is a unit
in $k$, hence the long exact sequence of the pair yields
$H^\bullet(X_K,R_-;k_\phi)=0$.
On the other hand, $\widehat{HFK}(K;\F_2)$ is nonzero:
its graded Euler characteristic is $\Delta_K(t)$, which is not
the zero polynomial (and has value $1$ at $t=1$).
Flat scalar extension from $\F_2$ to $k$ preserves nonzero
homology.  Two complexes with different cohomology cannot be
quasi-isomorphic, irrespective of a common shift.
\end{proof}

\begin{remark}[Scope of the chain comparison]
The contractions in Theorem~\ref{thm:chain-gluing} use the standard
mapping-cone and Mayer--Vietoris constructions.  The Floer tensor
pairing and its Euler-characteristic identity are established results.
Theorem~\ref{thm:chain-nogo} concerns the ordinary complexes specified
there; it does not exclude comparisons using additional boundary data.
\end{remark}

\subsection{A concrete comparison: two trefoils}
For the right-handed trefoil, write
$\mathrm{Eu}(\mathcal F_{3_1};t)=t^{-1}-1+t$.  Its connected sum
with itself gives
\[
\mathrm{Eu}(\mathcal F_{3_1\#3_1};t)
 =(t^{-1}-1+t)^2=t^{-2}-2t^{-1}+3-2t+t^2.
\]
The relative torsion has the same normalized Laurent expression.
On the \emph{absolute} cellular BF side, if
$\tau_{3_1}^{\mathrm{sym}}(t)=(t^{-1}-1+t)/(t-1)$, then
annular gluing produces
\[
(t-1)\bigl(\tau_{3_1}^{\mathrm{sym}}(t)\bigr)^2
=\frac{(t^{-1}-1+t)^2}{t-1},
\]
exactly the absolute torsion of the connected-sum complement.
The annular contracting homotopy $h_A=(t-1)^{-1}$ displays
\emph{at chain level} the same correction factor that is often
inserted only after forming scalar torsions.  The example
illustrates a precise common multiplicative shadow while
leaving the Maslov-graded Floer differential out of the BF
complex.

\section{L-space knots}
\label{sec:lspace}
The acyclicity obstruction of Theorem~\ref{thm:chain-nogo} rules
out a direct quasi-isomorphism of the \emph{standard} BF and Floer
complexes.  It does not prohibit using torsion as input to a
\emph{reconstruction algorithm} for a Floer complex when a
classification theorem makes this possible.  Positive L-space knots
provide a nontrivial class where this stronger kind of relationship
can be made completely explicit.

A knot $K\subset S^3$ is called a \emph{positive L-space knot} if
some positive integral Dehn surgery on $K$ is an L-space, i.e. a
rational homology sphere $Y$ satisfying
$\operatorname{rk}\widehat{HF}(Y)=|H_1(Y;\Z)|$.
The Ozsv\'ath--Szab\'o lens-space-surgery restriction and the
resulting staircase normal form for its knot Floer complex imply
that a positive L-space knot has a symmetrized Alexander polynomial
of the form
\begin{equation}\label{eq:lspace-alex}
 \Delta_K(t)=\sum_{i=-k}^{k}(-1)^{k-i}t^{n_i},\qquad
 n_{-k}<\cdots<n_k,\quad n_{-i}=-n_i,
\end{equation}
with integral exponents and $n_k=g(K)$.  The staircase
normal form is recalled below; for the classification input see
\cite{OzsvathSzaboLspace,Krcatovich}.

\begin{theorem}[Cyclotomic reconstruction of an L-space knot Floer \emph{chain complex}]
\label{thm:lspace-chain}
Fix $g\geq0$ and an odd prime $p\geq2g+1$, and let
$\zeta=\exp(2\pi i/p)$.  Suppose $K$ is a \emph{positive L-space
knot} with $g(K)\leq g$.  Knowing the single exact algebraic number
\[
 \alpha=\mathscr B_K(\zeta)
 =\tau(X_K,R_-;L_\zeta)^{\mathrm{sym}}=\Delta_K(\zeta),
\]
with the explicitly stipulated symmetric normalization,
determines the filtered chain homotopy type of
$CFK^-(K;\F_2)$ over $\F_2[U]$.  More precisely, the following
algorithm produces a filtered free chain complex representing it:
\begin{enumerate}[label=\textup{(\roman*)},leftmargin=2em]
\item For $-g\le a\le g$, compute the \emph{integer} coefficients
\begin{equation}\label{eq:lspace-trace}
 c_a=\frac{1+\Tr_{\Q(\zeta)/\Q}(\zeta^{-a}\alpha)}{p}.
\end{equation}
\item List the exponents for which $c_a\ne0$ in increasing order,
$\{n_{-k},\ldots,n_k\}$.  Under the positive L-space-knot
hypothesis, they satisfy \eqref{eq:lspace-alex}.
\item Let $S(\alpha)$ be free over $\F_2[U]$ on generators
$x_{-k},\ldots,x_k$ with Alexander filtration
$A(x_i)=n_i$ and $A(Ux_i)=n_i-1$.  Define its differential by
\begin{equation}\label{eq:lspace-staircase}
 \partial x_i=\begin{cases}
 0,&i\equiv k\pmod2,\\[2pt]
 x_{i-1}+U^{n_{i+1}-n_i}x_{i+1},
  &i\equiv k+1\pmod2.
 \end{cases}
\end{equation}
The second line only occurs at interior indices.  The usual
absolute Maslov normalization for $S^3$ determines the remaining
grading shift, with $M(Ux)=M(x)-2$.
\end{enumerate}
Then $S(\alpha)$ is filtered chain homotopy equivalent to
$CFK^-(K;\F_2)$.  Thus, for this class of knots, one exact
cyclotomic value of the normalized meridional BV--BFV torsion
character reconstructs more than the Euler characteristic: it
recovers the filtered Floer chain homotopy type and hence its
standard filtered-chain invariants.
\end{theorem}
\begin{proof}
By Theorem~\ref{thm:main} and
Proposition~\ref{prop:extract}, the exact cyclotomic value
$\alpha=\Delta_K(\zeta)$ together with the degree bound and
$\Delta_K(1)=1$ uniquely determines all coefficients via
\eqref{eq:lspace-trace}.  The L-space-knot condition now
supplies the alternating nonzero coefficients and ordered exponents
in \eqref{eq:lspace-alex} by the Ozsv\'ath--Szab\'o
classification theorem \cite{OzsvathSzaboLspace}.

Equation \eqref{eq:lspace-staircase} defines a differential.
Indeed, if $i\equiv k\pmod2$ then $\partial x_i=0$.
If $i\equiv k+1\pmod2$, its two neighbors have parity $k$,
so $\partial^2x_i=0$.  Moreover, the two terms in the
nonzero differential have Alexander filtrations
\[
 A(x_{i-1})=n_{i-1}<n_i,\qquad
 A(U^{n_{i+1}-n_i}x_{i+1})=n_i,
\]
so $\partial$ respects the increasing Alexander filtration.
The polynomial \eqref{eq:lspace-alex} specifies all the gaps
$n_{i+1}-n_i$ and hence all the nonzero differential entries.
The staircase normal-form theorem for positive L-space knots
(e.g. the basis-and-differential formulation in
\cite[Proposition~1.1]{Krcatovich}) identifies this complex,
with its $U$-action and Alexander filtration, with
$CFK^-(K;\F_2)$ up to filtered chain homotopy equivalence.
The usual Maslov grading conventions fix the overall shift.
\end{proof}

\begin{example}[The trefoil from one cubic holonomy value]
\label{ex:lspace-trefoil}
For the right-handed trefoil, $g=1$, and the minimal permitted
prime is $p=3$.  Since
$\Delta_{3_1}(t)=t^{-1}-1+t$, a primitive cube root
$\zeta_3$ gives $\alpha=\Delta_{3_1}(\zeta_3)=-2$.
Formula \eqref{eq:lspace-trace} returns
$c_{-1}=1$, $c_0=-1$, $c_1=1$.  Therefore $k=1$ and
$(n_{-1},n_0,n_1)=(-1,0,1)$.  The reconstructed
$\F_2[U]$-complex has generators $x_{-1},x_0,x_1$ and
\[
\partial x_{-1}=\partial x_1=0,
\qquad \partial x_0=x_{-1}+Ux_1.
\]
This is the positive-trefoil staircase model, not just its
Euler-characteristic polynomial.
\end{example}

\begin{example}[The cinquefoil]
For the positive $(2,5)$ torus knot, $g=2$; a primitive fifth
root $\zeta_5$ is sufficient.  The trace recovers the five
nonzero exponents $(-2,-1,0,1,2)$, and the resulting model has
\[
\partial x_{-1}=x_{-2}+Ux_0,\qquad
\partial x_1=x_0+Ux_2,
\]
with all other generator differentials zero.  In particular,
the two step lengths of the staircase are recovered from the
one specialized BV--BFV torsion character.
\end{example}

\begin{remark}[The L-space hypothesis]
The Alexander polynomial does not determine whether a knot is an
L-space knot.  Theorem~\ref{thm:lspace-chain} assumes this property and
uses the established staircase classification; it does not construct a
map from the twisted BF complex.  For general knots, the
Kinoshita--Terasaka example in Proposition~\ref{prop:nogo} obstructs
reconstruction from the Alexander polynomial.
\end{remark}

\section{Scalar extension and pairing obstructions}
\label{sec:categorical}

The preceding theorems compare torsion and Euler characteristic, but a
functor between boundary theories would have to intertwine their gluing
operations \emph{before} taking decategorification.  We now make the simplest
possible such attempt fully precise, prove that it cannot work, and isolate
which additional algebraic operations are indispensable.  The result is an
algebraic theorem about the standard torus algebra; it does not assume a
hypothetical equivalence of BV--BFV and Floer field theories.

\subsection{Derived complexes and torsion}
\label{subsec:levels}
For a knot exterior $X_K$ and nonresonant rank-one character $z$, the
finite-dimensional linear BF model produces a twisted relative complex
$C^\bullet(X_K,R_-;k_z)$ together with cellular bases, and hence a
Milnor-torsion element.  There are at least three distinguishable
constructions:
\begin{enumerate}[label=\textup{(\alph*)},leftmargin=2em]
\item the object $[C^\bullet]$ in the ordinary derived category $D(k)$;
\item the \emph{based} complex with a chosen determinant-line
      trivialization and its torsion;
\item a quantum BV--BFV state with boundary observable algebra,
      polarization, residual fields and higher composition operations.
\end{enumerate}
Our earlier torsion comparison concerns (b).  It gives no
identification of the full structure in (c) with a bordered Floer module.
Even (a) and (b) cannot be conflated:

\begin{lemma}[Acyclicity loses torsion]\label{lem:derived-loses-torsion}
Let $k$ be a field and $u,v\in k^\times$.  The two based cochain
complexes
\[
 E_u:\;0\longrightarrow k\xrightarrow{u} k\longrightarrow 0,
 \qquad
 E_v:\;0\longrightarrow k\xrightarrow{v} k\longrightarrow 0
\]
are isomorphic to the zero object in $D(k)$, but their based torsions
can have different values.  With the chain-torsion convention of
Definition~\ref{def:milnor}, their torsions differ by $u/v$ or its
inverse according to the cochain-to-chain convention.  In particular
there is no function on the isomorphism classes of $D(k)$ that
recovers based Reidemeister torsion for all such complexes.
\end{lemma}
\begin{proof}
Both differentials are invertible, so both complexes are contractible.
Their images in $D(k)$ are zero.  In the fixed standard bases the only
nontrivial determinant entering the acyclic torsion is the determinant
of the one differential; changing $u$ to $v$ changes this factor by
$u/v$ (possibly inverted).  Taking $u\ne v$ establishes the claim.
\end{proof}

This small example rules out an exact functor defined \emph{only} on the
ordinary twisted derived BF complex as a mechanism for remembering
Alexander torsion.  Simple-homotopy data, cellular lifts or a
suitable determinant-line enhancement are not optional decorations.
We next show that retaining only a scalar determinant line is still
insufficient for \emph{bordered} Floer gluing.

\subsection{The torus algebra and its chord ideal}
\label{subsec:torus-algebra}
We work initially over $\F_2$ and then extend scalars to an arbitrary
field $k$ of characteristic two, for instance $k=\F_2(t)$.  We use the
ordinary \emph{hat} bordered torus algebra of
Lipshitz--Ozsv\'ath--Thurston \cite[Section~10.1]{LOT};
the generators and relations are also recorded in
\cite[equation (1.18)]{LOTNotes}.  Its underlying
differential is zero and its eight-dimensional vector-space basis is
\begin{equation}\label{eq:torus-basis}
 \mathcal A=\F_2\langle
 \iota_0,\iota_1,
 \rho_1,\rho_2,\rho_3,\rho_{12},\rho_{23},\rho_{123}\rangle.
\end{equation}
The idempotents are orthogonal, $\iota_i\iota_j=\delta_{ij}\iota_i$
and $\iota_0+\iota_1=1$.  The Reeb-chord elements have endpoints
\[
\begin{array}{c|c}
 \iota_0\mathcal A\iota_0 & \iota_0,\rho_{12}\\
 \iota_0\mathcal A\iota_1 & \rho_1,\rho_3,\rho_{123}\\
 \iota_1\mathcal A\iota_0 & \rho_2\\
 \iota_1\mathcal A\iota_1 & \iota_1,\rho_{23}.
\end{array}
\]
The nonzero products of positive-length chords are determined by
\begin{equation}\label{eq:torus-products}
 \rho_1\rho_2=\rho_{12},\quad
 \rho_2\rho_3=\rho_{23},\quad
 \rho_1\rho_{23}=\rho_{12}\rho_3=\rho_{123},
\end{equation}
while $\rho_2\rho_1=\rho_3\rho_2=0$; remaining incompatible
products vanish.  We suppress the refined group grading in the
calculations below, regarding $\mathcal A_k:=\mathcal A\otimes_{\F_2}k$
as a differential-graded algebra concentrated in cohomological
degree zero.  No claim about an equivalence preserving the original
bordered grading is hidden in this simplification.

Let $J\subset\mathcal A_k$ be the two-sided ideal spanned by the
six Reeb-chord elements.  Equation~\eqref{eq:torus-products} gives
\begin{equation}\label{eq:radical}
 J^4=0,\qquad \mathcal A_k/J\cong k\iota_0\oplus k\iota_1\cong k\times k.
\end{equation}
In particular, $J$ is the Jacobson radical of this finite-dimensional
algebra.  The two quotient projections define strict unital
augmentations
\begin{equation}\label{eq:augmentations}
 \epsilon_i:\mathcal A_k\longrightarrow k,\qquad
 \epsilon_i(\iota_j)=\delta_{ij},\quad \epsilon_i(J)=0
 \qquad(i=0,1).
\end{equation}
Each augmentation makes $k$ into a left or right
$\mathcal A_k$-module, denoted $k_{\epsilon_i}$.

\begin{proposition}[Universality of rank-one strict scalarization]
\label{prop:augmentation-classification}
For every field $k$ of characteristic two, the only unital
$k$-algebra homomorphisms $\mathcal A_k\to k$ are
$\epsilon_0$ and $\epsilon_1$.  When both algebras are regarded
as dg algebras concentrated in cohomological degree zero, the
same statement holds for unital $A_\infty$ morphisms:
all higher Taylor maps $f_n$ for $n\ge2$ vanish for degree reasons,
and $f_1$ is one of the two augmentations.
\end{proposition}
\begin{proof}
Let $f:\mathcal A_k\to k$ be a unital algebra homomorphism.
For every $x\in J$, equation~\eqref{eq:radical} gives $x^4=0$,
so $f(x)^4=0$.  Since $k$ is a field, $f(x)=0$.  Thus $f$
factors through $\mathcal A_k/J\cong k\times k$.
The orthogonal idempotents $f(\iota_0)$ and $f(\iota_1)$ are
idempotents in a field whose sum is one.  Exactly one is one and
the other zero, proving the first assertion.  For an uncurved
$A_\infty$ morphism of dg algebras concentrated in degree zero,
$f_n$ has cohomological degree $1-n$, so $f_n=0$ for $n\ge2$.
The $n=2$ $A_\infty$ identity then says that $f_1$ respects
products, and strict unitality says $f_1(1)=1$.
\end{proof}

\begin{remark}[Keeping the actual bordered grading]
The preceding assertion about $A_\infty$ morphisms is deliberately
restricted to the \emph{degree-zero regrading} of the hat torus
algebra.  With the usual noncommutative group grading, or with
additional target cohomological degrees, higher maps need not
vanish for this reason.  Nevertheless, the first Taylor map on
cohomology of a strictly unital $A_\infty$ morphism to a
one-dimensional field target is an algebra character, and so
kills $J$.  This is an obstruction to recovering chord information
through the induced cohomology-algebra map; it is \emph{not}
a proof that all imaginable $A_\infty$ functors are impossible.
\end{remark}

\subsection{Derived scalar extension}
\label{subsec:reduction-functor}
The augmentations are not merely formal maps.  Each gives a
well-defined exact derived functor on the category of right modules,
\begin{equation}\label{eq:derived-functor}
 \mathsf R_i:\operatorname{Perf}(\mathcal A_k)
 \longrightarrow D(k),\qquad
 \mathsf R_i(M)=M\otimes^{\mathbf L}_{\mathcal A_k}k_{\epsilon_i}.
\end{equation}
Here $\operatorname{Perf}(\mathcal A_k)$ is the category of perfect
dg modules.  An $A_\infty$ right module may be replaced by a
bar-type dg model before applying the derived tensor product.
Consequently~\eqref{eq:derived-functor} is homotopy invariant and
respects all $A_\infty$ identities \emph{internal to this
extension-of-scalars construction}.  Its construction uses the torus algebra and does not assign a
BV--BFV state to a knot exterior.

For a left module we similarly write
$\mathsf L_i(N)=k_{\epsilon_i}\otimes^{\mathbf L}_{\mathcal A_k}N$.
A gluing-preserving rank-one reduction would require, for all
appropriate module pairs, a natural quasi-isomorphism between
the original derived pairing and the pairing of the reductions:
\begin{equation}\label{eq:desired-pairing}
 M\otimes^{\mathbf L}_{\mathcal A_k}N
 \overset{?}{\simeq}
 \mathsf R_i(M)\otimes_k^{\mathbf L}\mathsf L_i(N).
\end{equation}
The following theorem shows that equation~\eqref{eq:desired-pairing}
cannot hold on the entire bordered module category, even after
forgetting all gradings.

\begin{theorem}[Explicit Reeb-chord obstruction to scalar gluing]
\label{thm:pairing-defect}
Fix $i\in\{0,1\}$ and let $j=1-i$.  In the degree-zero model of
the bordered torus algebra over $k$, take the projective right
module $P_i=\iota_i\mathcal A_k$ and the projective left module
$Q_j=\mathcal A_k\iota_j$.  Then
\begin{equation}\label{eq:pairing-defect}
 H^0\!\left(P_i\otimes^{\mathbf L}_{\mathcal A_k}Q_j\right)
 \cong\iota_i\mathcal A_k\iota_j\ne0,
 \qquad
 \mathsf R_i(P_i)\otimes_k^{\mathbf L}\mathsf L_i(Q_j)
 \simeq0.
\end{equation}
For $i=0$ the first space has dimension three; for $i=1$ it
has dimension one.  Hence no natural quasi-isomorphisms
\eqref{eq:desired-pairing} exist for all perfect modules.
In particular, \emph{scalar character base change does not
preserve the bordered gluing operation}.
\end{theorem}
\begin{proof}
Since $P_i$ and $Q_j$ are summands of free right and left
$\mathcal A_k$-modules, respectively, they are projective.
The derived tensor product is therefore represented by the
ordinary tensor product, and multiplication induces the
isomorphism
\[
 (\iota_i\mathcal A_k)\otimes_{\mathcal A_k}
 (\mathcal A_k\iota_j)\xrightarrow{\;\cong\;}
 \iota_i\mathcal A_k\iota_j.
\]
By the table in Section~\ref{subsec:torus-algebra},
$\iota_0\mathcal A_k\iota_1$ has basis
$\rho_1,\rho_3,\rho_{123}$, while
$\iota_1\mathcal A_k\iota_0$ has basis $\rho_2$.
On the other hand, projectivity and the augmentation identities give
\[
 \mathsf R_i(P_i)\cong\iota_i k_{\epsilon_i}\cong k,
 \qquad
 \mathsf L_i(Q_j)\cong k_{\epsilon_i}\iota_j=0
 \quad(j\ne i).
\]
The right-hand side of~\eqref{eq:pairing-defect} is therefore zero.
A nonzero complex concentrated in degree zero cannot be
quasi-isomorphic to the zero complex.
\end{proof}

\begin{corollary}[Two idempotents without Reeb chords do not suffice]
\label{cor:two-idempotents}
Replacing $\mathcal A_k$ by its semisimple quotient
$\mathcal A_k/J\cong k\times k$ also destroys the off-diagonal
pairings of Theorem~\ref{thm:pairing-defect}.  Thus even a boundary
model carrying \emph{both} scalar idempotent sectors cannot reproduce
all bordered tensor products if it discards the Reeb-chord ideal.
\end{corollary}
\begin{proof}
The quotient maps each $\iota_i\mathcal A_k\iota_j$ for $i\ne j$
to $\iota_i(\mathcal A_k/J)\iota_j=0$, although the original
space is nonzero.
\end{proof}

\begin{remark}[Interpretation and precise domain of the theorem]
The modules $P_i$ and $Q_j$ are elementary \emph{algebraic}
projective modules; we do not assert that each is the bordered
Floer invariant of a particular knot exterior.  The theorem
rules out a rank-one, gluing-preserving scalar reduction on the
\emph{whole perfect module category}; it does not rule out
specialized comparisons for particular geometric objects.
It also does not identify the derived functor $\mathsf R_i$ with
the BFV pushforward or with a physical choice of polarization.
Establishing such an identification would require a separate
construction of a boundary observable algebra and a functor
from it to $\mathcal A_k$.
\end{remark}

\subsection{The projective pairing matrix}
\label{subsec:rank-obstruction}
The preceding counterexample uses a specified augmentation.  There is a
stronger obstruction that does not assume that the scalar reductions
come from algebra maps, or even that the reductions are functorial.
Write $P_a=\iota_a\mathcal A_k$ and $Q_b=\mathcal A_k\iota_b$ for
$a,b\in\{0,1\}$.  Their derived pairing is concentrated in degree zero,
with dimension matrix
\begin{equation}\label{eq:projective-pairing-matrix}
 C_{ab}=\dim_k(\iota_a\mathcal A_k\iota_b),\qquad
 C=\begin{pmatrix}2&3\\1&2\end{pmatrix},\qquad \det C=1.
\end{equation}
The determinant in this equation is a determinant of
\emph{integer dimension counts}; it is not a computation in the
characteristic-two coefficient field.

\begin{theorem}[Obstruction to independent scalar tensor factorization]
\label{thm:scalar-rank-obstruction}
There do not exist four perfect complexes $V_0,V_1,W_0,W_1$ over $k$
such that, for every $a,b\in\{0,1\}$,
\begin{equation}\label{eq:scalar-factorization}
 P_a\otimes_{\mathcal A_k}^{\mathbf L}Q_b
 \simeq V_a\otimes_k^{\mathbf L}W_b.
\end{equation}
Consequently no pair of assignments from all perfect right and left
$\mathcal A_k$-modules to perfect scalar complexes can preserve every
derived pairing by ordinary tensor product over $k$.
The obstruction already occurs on the four projective pairs and requires
neither naturality nor an augmentation hypothesis.
\end{theorem}
\begin{proof}
For a perfect scalar complex $V$, put
$b(V)=\sum_n\dim_k H^n(V)$, its finite total homology dimension.
The K\"unneth formula over a field gives
$b(V\otimes_k^{\mathbf L}W)=b(V)b(W)$.
If~\eqref{eq:scalar-factorization} held, then
\[
 C_{ab}=b(V_a)b(W_b).
\]
The resulting $2\times2$ matrix is an outer product, so it has determinant
zero over the integers.  This contradicts~\eqref{eq:projective-pairing-matrix}.
\end{proof}

\begin{remark}[Scope of the rank obstruction]
This theorem concerns independent scalar complexes with the
\emph{ordinary tensor product} as pairing.  It does not exclude a target
with an additional boundary kernel, a non-scalar observable algebra, or
a pairing other than ordinary tensor product.  It also uses the same
degree-zero algebra model as Theorem~\ref{thm:pairing-defect}; no
extension to every grading or every geometric subcategory is asserted.
The rank computation is stronger than the augmentation example, but
it is not a general nonexistence theorem for BV--BFV/Floer functors.
\end{remark}

\subsection{Semisimple reduction and its pairing defect}
\label{subsec:radical-defect}
Let $S=\mathcal A_k/J=k\iota_0\oplus k\iota_1$, and denote the
quotient map by $\pi:\mathcal A_k\to S$.  For perfect right and left
modules define
\[
 R_S(M)=M\otimes_{\mathcal A_k}^{\mathbf L}S,\qquad
 L_S(N)=S\otimes_{\mathcal A_k}^{\mathbf L}N.
\]
These are perfect $S$-complexes.  Since $S$ is a product of two fields,
each consists of the two corresponding scalar complexes.  Keeping both
sectors has a useful positive property which must be distinguished from
preservation of the pairing.

\begin{proposition}[The two-sector reduction is conservative]
\label{prop:two-sector-conservative}
The functors $R_S$ and $L_S$ on perfect modules reflect
quasi-isomorphisms.  Equivalently, $R_S(M)\simeq0$ implies $M\simeq0$,
and the analogous assertion holds on the left.
Nevertheless, $R_S$ is not faithful on morphisms, and neither reduction
preserves all derived pairings.
\end{proposition}
\begin{proof}
Each bimodule $J^r/J^{r+1}$, for $1\le r\le3$, is annihilated by
$J$ on both sides, and hence is an $S$-bimodule.  Associativity of derived
tensor product gives
\[
 M\otimes_{\mathcal A_k}^{\mathbf L}(J^r/J^{r+1})
 \simeq R_S(M)\otimes_S^{\mathbf L}(J^r/J^{r+1}),
\]
where the right side is viewed as a right $\mathcal A_k$-complex by
restriction along $\pi$.  If $R_S(M)\simeq0$, all these complexes vanish.
Starting with $J^4=0$ and using the triangles associated to
$0\to J^{r+1}\to J^r\to J^r/J^{r+1}\to0$, downward induction gives
$M\otimes_{\mathcal A_k}^{\mathbf L}J\simeq0$.
The triangle obtained from $0\to J\to\mathcal A_k\to S\to0$
then gives $M\simeq0$.  Applying this argument to the cone of a
morphism proves reflection of quasi-isomorphisms.  The left-module
proof is identical.

For nonfaithfulness, consider the nonzero right-module morphism
\[
 f:P_1\longrightarrow P_0,\qquad
 f(\iota_1a)=\rho_1a.
\]
It is nonzero in the derived category because both modules are projective
and concentrated in degree zero.  Its reduction is zero since
$\pi(\rho_1)=0$.  Failure to preserve pairings follows already from
Corollary~\ref{cor:two-idempotents}, and is quantified below.
\end{proof}

\begin{theorem}[Canonical defect and its three chord layers]
\label{thm:radical-defect}
For perfect right and left $\mathcal A_k$-modules $M,N$, the quotient
$\pi$ induces a natural pairing map
\begin{equation}\label{eq:semisimple-pairing-map}
 \theta_{M,N}:M\otimes_{\mathcal A_k}^{\mathbf L}N
 \longrightarrow R_S(M)\otimes_S^{\mathbf L}L_S(N).
\end{equation}
Its homotopy fiber is
\begin{equation}\label{eq:radical-fiber}
 \mathscr D(M,N)=
 M\otimes_{\mathcal A_k}^{\mathbf L}J
   \otimes_{\mathcal A_k}^{\mathbf L}N.
\end{equation}
In particular there is a natural distinguished triangle
\begin{equation}\label{eq:radical-triangle}
 \mathscr D(M,N)\longrightarrow
 M\otimes_{\mathcal A_k}^{\mathbf L}N
 \xrightarrow{\theta_{M,N}}
 R_S(M)\otimes_S^{\mathbf L}L_S(N)
 \longrightarrow\mathscr D(M,N)[1].
\end{equation}
The defect admits a finite filtration with three quotients
\begin{equation}\label{eq:radical-graded-pieces}
 \operatorname{gr}^{r}\mathscr D(M,N)\simeq
 R_S(M)\otimes_S^{\mathbf L}V_r
       \otimes_S^{\mathbf L}L_S(N),\qquad
 V_r=J^r/J^{r+1},\quad 1\le r\le3.
\end{equation}
Representatives for the three $S$-bimodules and their endpoint
dimension matrices $G^{(r)}_{ab}=\dim_k(\iota_aV_r\iota_b)$ are
\begin{equation}\label{eq:chord-layer-matrices}
\begin{array}{c|c|c}
 r & \text{basis representatives for }V_r & G^{(r)}\\ \hline
 1 & \rho_1,\rho_2,\rho_3 &
       \begin{pmatrix}0&2\\1&0\end{pmatrix}\\[4pt]
 2 & \rho_{12},\rho_{23} &
       \begin{pmatrix}1&0\\0&1\end{pmatrix}\\[4pt]
 3 & \rho_{123} &
       \begin{pmatrix}0&1\\0&0\end{pmatrix}
\end{array}
\end{equation}
Thus the associated finite spectral sequence computes the homology
of the pairing defect from these layers and their extension maps.
For projective pairs,
\begin{equation}\label{eq:projective-defect}
 \mathscr D(P_a,Q_b)\cong\iota_aJ\iota_b,\qquad
 \bigl(\dim_k\mathscr D(P_a,Q_b)\bigr)_{ab}
   =C-I=\begin{pmatrix}1&3\\1&1\end{pmatrix}.
\end{equation}
The quotient map~\eqref{eq:semisimple-pairing-map} is a quasi-isomorphism
for a given pair exactly when its defect is acyclic.
\end{theorem}
\begin{proof}
The short exact sequence
$0\to J\to\mathcal A_k\xrightarrow{\pi}S\to0$
is a short exact sequence of $\mathcal A_k$-bimodules.  Derived tensoring
on the left with $M$ and on the right with $N$ gives a distinguished
triangle.  Its middle term is $M\otimes_{\mathcal A_k}^{\mathbf L}N$.
Associativity identifies its last term as
\[
 M\otimes_{\mathcal A_k}^{\mathbf L}S
   \otimes_{\mathcal A_k}^{\mathbf L}N
 \simeq
 (M\otimes_{\mathcal A_k}^{\mathbf L}S)
 \otimes_S^{\mathbf L}
 (S\otimes_{\mathcal A_k}^{\mathbf L}N).
\]
This proves~\eqref{eq:radical-triangle} and the fiber assertion.

The multiplication table gives
\[
 J^2=\langle\rho_{12},\rho_{23},\rho_{123}\rangle_k,\qquad
 J^3=\langle\rho_{123}\rangle_k,\qquad J^4=0.
\]
Choose bounded complexes of finitely generated projective modules
representing $M$ and $N$.  Tensoring them with the filtration
$J\supset J^2\supset J^3\supset0$ gives a filtration of a representative
of~\eqref{eq:radical-fiber}; projectivity preserves the short exact
sequences in each degree.  Its quotients are
$M\otimes_{\mathcal A_k}^{\mathbf L}V_r
  \otimes_{\mathcal A_k}^{\mathbf L}N$.
Since the two actions on $V_r$ factor through $S$, associativity
identifies these with~\eqref{eq:radical-graded-pieces}.
The endpoint table for the chords proves~\eqref{eq:chord-layer-matrices}.
The filtration is finite and the representative complexes are bounded
and finite-dimensional over $k$, so its spectral sequence converges
to $H^\bullet(\mathscr D(M,N))$.

For $M=P_a,N=Q_b$, no resolutions are needed.  Multiplication identifies
$P_a\otimes_{\mathcal A_k}J\otimes_{\mathcal A_k}Q_b$
with $\iota_aJ\iota_b$, yielding~\eqref{eq:projective-defect}.
The last assertion follows directly from the distinguished triangle.
\end{proof}

\begin{corollary}[Euler pairing and its exact scalar defect]
\label{cor:Euler-defect}
Let
\[
 m_a=\chi(R_S(M)\iota_a),\qquad
 n_b=\chi(\iota_bL_S(N)),
\]
and regard $m,n$ as two-component columns of integers.  Then
\begin{align}
 \chi(M\otimes_{\mathcal A_k}^{\mathbf L}N)&=m^{\mathsf T}Cn,
       \label{eq:Cartan-Euler}\\
 \chi(R_S(M)\otimes_S^{\mathbf L}L_S(N))&=m^{\mathsf T}n,
       \label{eq:semisimple-Euler}\\
 \chi(\mathscr D(M,N))&=m^{\mathsf T}(C-I)n.
       \label{eq:defect-Euler}
\end{align}
Here the Euler characteristics refer to the ordinary cohomological
grading of the specified degree-zero algebra model.
\end{corollary}
\begin{proof}
Euler characteristic is additive under the finite filtration and
under~\eqref{eq:radical-triangle}, and multiplicative for tensor
products of finite scalar complexes.  The decomposition
$S=k\iota_0\oplus k\iota_1$ therefore gives
\[
 \chi(\operatorname{gr}^{r}\mathscr D(M,N))
       =m^{\mathsf T}G^{(r)}n.
\]
The matrices satisfy $I+G^{(1)}+G^{(2)}+G^{(3)}=C$.
Adding the three layer contributions proves the formulas.
\end{proof}

\begin{remark}[What the filtration does and does not reconstruct]
Conservativity detects whether a perfect object or a morphism cone is
acyclic; it is weaker than faithfulness or equivalence of categories.
Likewise, the matrices~\eqref{eq:chord-layer-matrices} determine the
Euler defect, but do not determine the extension maps in the filtered
defect complex.  No direct-sum splitting of
\eqref{eq:radical-graded-pieces} is asserted for general $M,N$.
In particular, the dimensions of the scalar sectors alone do not
reconstruct the original pairing.  The nontrivial products
$\rho_1\rho_2=\rho_{12}$,
$\rho_2\rho_3=\rho_{23}$ and
$\rho_1\rho_2\rho_3=\rho_{123}$
are precisely the algebraic information omitted by the semisimple
quotient.

This calculation supplies a concrete condition on a proposed enriched
boundary comparison: if it uses $\pi$ as its reduction, then its
additional gluing data must represent the bimodule $J$ and the derived
defect~\eqref{eq:radical-fiber}, rather than simply retaining two
independent torsion lines.  It is an algebraic requirement; a realization
by quantum BV--BFV states, in characteristic zero and with the original
bordered grading, would still require a separate construction.
\end{remark}

\subsection{Torsion and object reconstruction}
\label{subsec:alextrivial}
Theorem~\ref{thm:pairing-defect} is an obstruction to preserving
\emph{composition}.  There is a separate obstruction to recovering
\emph{objects} from the scalar torsion function, even if all
nonresonant characters are retained.  Define the normalized
meridional BF torsion function
\[
 \mathscr T_K(t)=\tau(X_K,R_-;\Q(t)_\phi)^{\rm sym}
 =\Delta_K(t).
\]
Here the equality uses the chosen symmetric normalization and
means the particular relative Milnor-torsion representative used
throughout this paper.

\begin{theorem}[Failure of object reconstruction from all scalar twists]
\label{thm:all-characters}
There are knots $K$ and $K'$ with
$\mathscr T_K(t)=\mathscr T_{K'}(t)$ as rational functions,
and hence with equal nonresonant specializations at \emph{every}
character, but whose Alexander-graded knot Floer homologies
are not isomorphic.  Consequently no rule assigning a Floer
object \emph{solely as a function of the normalized scalar
meridional BF torsion} can recover knot Floer homology for all knots.
\end{theorem}
\begin{proof}
Let $K$ be the unknot and $K'$ the Kinoshita--Terasaka knot.
Both have Alexander polynomial $1$, so their relative torsion
functions and every specialization coincide.  The unknot has
Seifert genus zero, whereas the Kinoshita--Terasaka knot has
genus two.  Knot Floer homology detects Seifert genus
\cite{OzsvathSzaboGenus,OzsvathSzaboMutation}; hence
$\widehat{HFK}(K)$ and $\widehat{HFK}(K')$ have different
Alexander support and are nonisomorphic.
\end{proof}

\subsection{The resulting comparison diagram}
The results above give the following diagram of \emph{known}
operations and obstructions.  The upper row is a shadow identity;
there is no implied upper-to-lower equivalence of categories:
\[
\begin{tikzcd}[column sep=large,row sep=large]
 \text{based relative BF complexes}
   \arrow[r,"\text{torsion}"]
 & \Q(t)^\times/\{\pm t^n\}
 & \text{knot Floer complexes}
   \arrow[l,swap,"\chi_A"]\\
 \operatorname{Perf}(\mathcal A_k)
   \arrow[r,swap,"-\otimes^{\mathbf L}_{\mathcal A_k}k_{\epsilon_i}"]
 & D(k)
 \end{tikzcd}
\]
The top horizontal comparison holds on the geometrically
specified knot data with the normalizations of this paper.
The lower arrow is an actual derived, $A_\infty$-compatible
base-change functor.  Theorem~\ref{thm:pairing-defect} shows
why it is \emph{not} a tensor/gluing bridge to a scalar boundary
theory.  Lemma~\ref{lem:derived-loses-torsion} says why its
ordinary derived output cannot even retain based torsion without
additional structure.  Theorem~\ref{thm:all-characters} shows
that scalar torsion is insufficient to reconstruct all knot
Floer objects.  The distinction between these statements is
important: they rule out three particular modes of comparison,
not every potential enriched BV--BFV/Floer correspondence.

\subsection{Requirements for an enriched comparison}
The obstruction suggests a more sharply formulated research
problem.  Rather than seeking a functor from one-dimensional
BFV torsion lines to bordered modules, one would need an
enhancement $\mathcal B_{\partial}$ of the boundary theory,
containing at least nonzero off-diagonal morphisms that can model
the Reeb-chord sector, together with a homotopy-coherent
comparison with the bordered torus algebra.  Such an enhancement
should respect the knot's peripheral parametrization, derived
pairings, gradings and the relative determinant-line
normalization.  The existence of a geometrically natural
BV--BFV construction with all these properties remains open
within the present work; we do \emph{not} claim to have
constructed it.

\section{A dg resolution of the torus algebra}
\label{sec:constructive}
The scalar obstructions do not prevent a comparison that retains the chord
algebra.  We now construct a small dg resolution of that algebra and prove
a pairing-preserving equivalence of perfect-module categories.
We also realize the differential on its framed representation space as a
classical BFV Hamiltonian and as the target of a one-dimensional BV--BFV
cotangent model.  These are two explicit constructions.  They must not
be conflated with a derivation of the resolved algebra from the
three-dimensional BF theory on a knot exterior: that identification is
not supplied by choosing a resolution of the Floer algebra.

\subsection{A six-generator resolution}
\label{subsec:six-generator}
In this subsection $k$ is any field.  Let $A_k$ be the algebra with the
same generators, idempotents and unsigned multiplication table as the
degree-zero torus algebra.  In characteristic two it is the
$\mathcal A_k$ of Section~\ref{sec:categorical}.  In other characteristics
it denotes an algebraic lift of that table, not an assertion about
sign-refined bordered Floer invariants.  Set
$a=\rho_1$, $b=\rho_2$, $c=\rho_3$; thus
\[
 a,c\in\iota_0A_k\iota_1,\qquad
 b\in\iota_1A_k\iota_0,\qquad ba=cb=0.
\]
The multiplication convention reads paths from left to right, as in
\cite[equation (1.18)]{LOTNotes}.  Define the free graded path algebra
$B_k$ over $S=k\iota_0\oplus k\iota_1$ by six generators:
\begin{equation}\label{eq:resolved-generators}
\begin{array}{c|c|c}
 \text{generator} & \text{cohomological degree} & \text{idempotent endpoints}\\ \hline
 a,c & 0 & \iota_0(-)\iota_1\\
 b   & 0 & \iota_1(-)\iota_0\\
 u   & -1 & \iota_1(-)\iota_1\\
 v   & -1 & \iota_0(-)\iota_0\\
 w   & -2 & \iota_0(-)\iota_1
\end{array}
\end{equation}
There are no relations besides the endpoint rules.  In particular $B_k$
has finitely many generators but is infinite-dimensional.  Its differential
is the graded derivation determined by
\begin{equation}\label{eq:resolved-differential}
 da=db=dc=0,\qquad du=ba,\qquad dv=cb,\qquad dw=cu-va.
\end{equation}
The generators $u,v$ resolve the two relations, and $w$ resolves their
overlap.  There is an augmented dg-algebra map
\begin{equation}\label{eq:resolved-counit}
 q:B_k\longrightarrow A_k,\qquad
 q(a)=\rho_1,\quad q(b)=\rho_2,\quad q(c)=\rho_3,\quad
 q(u)=q(v)=q(w)=0.
\end{equation}

\begin{theorem}[Explicit dg resolution of the torus algebra]
\label{thm:small-resolution}
The differential~\eqref{eq:resolved-differential} squares to zero, and
$q$ is a quasi-isomorphism.  More precisely,
\[
 H^n(B_k)=0\quad(n\ne0),\qquad
 H^0(B_k)=
 k\langle\iota_0,\iota_1,a,b,c,ab,bc,abc\rangle\cong A_k.
\]
Every expanded path containing a forbidden subword admits an explicit
contracting homotopy.  Thus the assertion follows from a direct
word-by-word contraction, without assuming Koszul formality.
\end{theorem}
\begin{proof}
The only generator with a potentially nonzero second differential is $w$,
and
\[
 d^2w=c\,du-(dv)a=cba-cba=0.
\]
Since $d$ is a graded derivation, this proves $d^2=0$.

For a monomial in $B_k$, form its \emph{expanded path} by replacing
$u$ by $ba$, $v$ by $cb$, and $w$ by $cba$.  The differential preserves
the expanded path.  Hence the underlying complex of $B_k$ is the
direct sum of the complexes associated to individual expanded paths
in the alphabet $a,b,c$, together with the two idempotents.

Mark each occurrence of $ba$ or $cb$ in an expanded path.  Two marked
occurrences overlap only as the pair $cb,ba$ inside $cba$.
There is no chain of three overlapping occurrences, since the letters
that finish $ba$ and begin $cb$ are incompatible with extending such a
chain.  The marked occurrences therefore split into disjoint clusters
of two types: a single relation word $ba$ or $cb$, or the overlapping
word $cba$.  The remaining intervals consist of ordinary arrows.

The complex for a single relation cluster is
$ku\xrightarrow{d}kba$, or $kv\xrightarrow{d}kcb$, with the first
term in degree $-1$.  It contracts by $h(ba)=u$, respectively
$h(cb)=v$, and $h(u)=h(v)=0$.
For an overlap cluster the complex is
\[
 kw\xrightarrow{d} k(cu)\oplus k(va)
   \xrightarrow{d} k(cba),
\]
in degrees $-2,-1,0$, with
$dw=cu-va$ and $d(cu)=d(va)=cba$.
An explicit contraction is
\begin{equation}\label{eq:overlap-contraction}
 h(cba)=cu,\qquad h(cu)=0,\qquad h(va)=-w,\qquad h(w)=0.
\end{equation}
Direct substitution gives $dh+hd=1$ and $h^2=0$.
The formulas remain valid in characteristic two, where the minus
sign becomes a plus sign.

A basis monomial with a fixed expanded path is exactly a tiling of
that path by its arrows and the generators $u,v,w$.  Such tilings
factor over the disjoint clusters just described.  The corresponding
complex is therefore a tensor product of the cluster complexes, with
the ordinary Koszul signs, and one-dimensional degree-zero factors
for the remaining arrow intervals.  If a relation cluster is present,
apply its contraction to the first such factor and the identity to
the others.  This contracts the whole expanded-path complex.

An expanded path without a forbidden subword contributes precisely one
degree-zero basis element.  The only such nonempty paths are
$a,b,c,ab,bc,abc$: a compatible path of length four necessarily
contains $ba$ or $cb$.  Including the two idempotents gives the eight
listed cohomology classes.  The map $q$ identifies their products
with the torus multiplication table and annihilates all contracted
blocks, proving the quasi-isomorphism.
\end{proof}

\begin{remark}[Why the coherence generator is indispensable]
If one adds only $u,v$ with $du=ba,dv=cb$, the element $cu-va$
is a degree-$-1$ cycle.  Its expanded-path block has no degree-$-2$
generator, so it is not a boundary.  Such a two-relation construction
is not a dg resolution of $A_k$.  Adding $w$ eliminates exactly this
overlap class; Theorem~\ref{thm:small-resolution} proves that no further
generators are needed.  This is the concrete coherence that is lost
when the chords are replaced by scalar idempotents.
\end{remark}

\subsection{Derived equivalence and tensor pairings}
\label{subsec:resolved-equivalence}
\begin{theorem}[Resolved boundary-category equivalence]
\label{thm:resolved-equivalence}
Extension and restriction of scalars along $q$ give inverse equivalences
\begin{equation}\label{eq:resolved-Morita}
 F=-\otimes_{B_k}^{\mathbf L}A_k:
       \operatorname{Perf}(B_k)
       \ \rightleftarrows\
       \operatorname{Perf}(A_k):G=q^* .
\end{equation}
For perfect right $A_k$-modules $M$ and perfect left $A_k$-modules $N$
there is a natural quasi-isomorphism
\begin{equation}\label{eq:resolved-pairing}
 q^*M\otimes_{B_k}^{\mathbf L}q^*N
       \simeq M\otimes_{A_k}^{\mathbf L}N.
\end{equation}
The same scalar-extension comparison respects derived composition of
perfect bimodules.  Thus the resolved boundary category retains the
entire algebraic gluing operation, rather than only its Euler shadow.
\end{theorem}
\begin{proof}
The quasi-isomorphism $q:B_k\to q^*A_k$ identifies the restricted free
$A_k$-module with the free $B_k$-module in the derived category.
The unit and counit of scalar extension and restriction are therefore
quasi-isomorphisms on free generators.  They remain quasi-isomorphisms
under shifts, finite cones and direct summands.  These operations generate
the perfect categories, which proves~\eqref{eq:resolved-Morita}.
This is a concrete instance of standard dg derived Morita theory
\cite{KellerDG}.

Likewise, the multiplication map
$A_k\otimes_{B_k}^{\mathbf L}A_k\to A_k$
is a quasi-isomorphism: tensor the quasi-isomorphism
$B_k\to A_k$ with a semifree replacement and use the derived unit.
Associativity then gives
\[
 q^*M\otimes_{B_k}^{\mathbf L}q^*N
 \simeq
 M\otimes_{A_k}^{\mathbf L}
 (A_k\otimes_{B_k}^{\mathbf L}A_k)
 \otimes_{A_k}^{\mathbf L}N
 \simeq M\otimes_{A_k}^{\mathbf L}N.
\]
Applying the same argument to bimodules gives compatibility with
composition; the associativity is the usual associativity of derived
tensor products, as realized by semifree or bar models.
\end{proof}

For the four projective pairs, equation~\eqref{eq:resolved-pairing}
recovers all entries of the matrix $C$ in
\eqref{eq:projective-pairing-matrix}, including its off-diagonal entries.
There is no conflict with Theorem~\ref{thm:scalar-rank-obstruction}:
the target category uses the non-scalar dg algebra $B_k$ and keeps
the relation and overlap homotopies.
Theorem~\ref{thm:resolved-equivalence} is an equivalence with the
\emph{specified algebraic torus-module model}.  We have not changed
the original grading qualifications or proved that these module
categories are the quantum state categories of Abelian BF theory.

\subsection{A classical BFV Hamiltonian for the resolution}
\label{subsec:cotangent-BFV}
Here work over a characteristic-zero field, for instance $\Q$, with
ordinary graded symplectic geometry.  Fix finite vector spaces
$V_0,V_1$ and consider the graded affine space $\mathcal R$ of framed
matrix coordinates for the six generators:
\[
\begin{array}{c|c}
 \text{coordinates} & \text{matrix spaces}\\ \hline
 a,c,w & \operatorname{Hom}(V_1,V_0)\\
 b & \operatorname{Hom}(V_0,V_1)\\
 u & \operatorname{End}(V_1)\\
 v & \operatorname{End}(V_0).
\end{array}
\]
Their degrees are those in~\eqref{eq:resolved-generators}.
The entries of the odd coordinates $u,v$ are graded-commuting odd
variables.  The degree-one vector field $Q_{\mathcal R}$ is
\begin{equation}\label{eq:matrix-BRST}
 Q_{\mathcal R}a=Q_{\mathcal R}b=Q_{\mathcal R}c=0,\quad
 Q_{\mathcal R}u=ba,\quad Q_{\mathcal R}v=cb,\quad
 Q_{\mathcal R}w=cu-va.
\end{equation}
The same overlap cancellation proves $Q_{\mathcal R}^2=0$.
This is the commutative coordinate dg algebra of the framed
representation construction for $B_\Q$.  At an ordinary degree-zero
point, $u,v,w$ vanish and the $Q$-fixed-point equations are $ba=cb=0$.
It is a framed representation model; no quotient by changes of basis
and no category of quantized branes is implicit in this statement.

Let $\mathcal P=T^*\mathcal R$, with its degree-zero canonical
symplectic form and Liouville one-form $\lambda$.
For a coordinate $x$, its momentum $p_x$ has degree $-|x|$.
Using the canonical graded evaluation convention, define
\begin{equation}\label{eq:cotangent-charge}
 \Theta=\langle p,Q_{\mathcal R}x\rangle
   =\langle p_u,ba\rangle+\langle p_v,cb\rangle
     +\langle p_w,cu-va\rangle .
\end{equation}
Each term has degree one.  The pairings are the trace pairings between
rectangular matrix spaces and their duals, with the Koszul convention
fixed by the canonical cotangent primitive.

\begin{proposition}[Classical cotangent BV--BFV realization]
\label{prop:cotangent-BVBFV}
The Hamiltonian vector field of $\Theta$ is the cotangent lift of
$Q_{\mathcal R}$ and is cohomological.  In particular,
$\{\Theta,\Theta\}=0$.  For an oriented interval $I$, the fields
\[
 \mathcal F_I=\operatorname{Map}(T[1]I,\mathcal P)
\]
carry the classical one-dimensional AKSZ BV--BFV structure with
bulk symplectic form and action
\begin{equation}\label{eq:interval-AKSZ}
 \omega_I=\int_{T[1]I}\delta\mathsf p\,\delta\mathsf x,\qquad
 S_I=\int_{T[1]I}
     \bigl(\langle\mathsf p,D\mathsf x\rangle
                  +\Theta(\mathsf x,\mathsf p)\bigr),
\end{equation}
where $D$ is the source de Rham vector field.
The degrees of $\omega_I$ and $S_I$ are $-1$ and $0$.
The boundary data are $\mathcal P$ at each endpoint, with outward
orientation signs:
\[
 \alpha_{\partial I}=\lambda_1-\lambda_0,\qquad
 \omega_{\partial I}=\delta\alpha_{\partial I},\qquad
 S_{\partial I}=\Theta_1-\Theta_0 .
\]
With the corresponding canonical sign conventions the bulk and
boundary structures satisfy the BV--BFV compatibility identity.
\end{proposition}
\begin{proof}
For any vector field $Q$ on a graded manifold, the fiber-linear
Hamiltonian $\langle p,Qx\rangle$ generates its cotangent lift.
The cotangent lift respects graded Lie brackets; hence its square
vanishes because $Q_{\mathcal R}^2=0$.
Equivalently, the Poisson self-bracket of this Hamiltonian is the
fiber-linear Hamiltonian for $[Q_{\mathcal R},Q_{\mathcal R}]$, which
is zero.  This proves the target master equation.

Transgression over $T[1]I$ lowers degree by one.
Both the kinetic term and $\Theta$ have target total degree one,
giving the displayed bulk degrees.  Varying the kinetic term and
integrating the source differential by parts yields the endpoint
Liouville forms with opposite orientation signs.  The bulk
cohomological vector field is the sum of $D$ and the lifted target
field; its square is zero by the source differential, the target
master equation and their graded commutation.  The integration-by-parts
identity is precisely the classical BV--BFV bulk--boundary
compatibility relation.  This is the standard AKSZ-with-boundary
construction applied to the explicitly specified target
\cite{CMRClassical}.
\end{proof}

\begin{remark}[Coefficient and dimension restrictions]
The algebraic resolution and module equivalence hold over $\F_2$
and over $\Q$, separately.  The classical graded symplectic construction
above is made over characteristic zero.  It is not a proof of a
change-of-coefficients theorem for bordered Floer homology.
Moreover, the bulk source here is an interval and the Hamiltonian is
defined from the torus relations.  It is a
\emph{one-dimensional classical BV--BFV cotangent model}; it is not
the original three-dimensional Abelian BF action.
Its quantum BFV state complex, if constructed, would also not be
identified with $\operatorname{Perf}(B_k)$ merely by the classical
Hamiltonian formula.
\end{remark}

\subsection{Geometric interpretation}
\label{subsec:geometric-gap}
The construction now supplies a concrete dg algebra and a classical
BV--BFV realization, as well as an unconditional algebraic equivalence
\eqref{eq:resolved-Morita}.  A full BV--BFV/bordered-Floer equivalence
still requires the following independent geometric work:
\begin{enumerate}[label=\textup{(\roman*)},leftmargin=2em]
 \item Construct a quantum boundary category from the chosen
 three-dimensional field theory and compute a generator's derived
 endomorphism algebra.  Prove that it is quasi-isomorphic to $B_k$
 (or derived Morita equivalent to it), using that field theory's
 observables and gauge fixing rather than inserting $A_k$ as input.
 \item For each parametrized knot exterior or cobordism, construct its
 boundary module from the same theory and identify it with the
 appropriate bordered Floer invariant.  The abstract equivalence of
 module categories does not identify which module a manifold supplies.
 \item Identify the field-theoretic sewing map with derived tensor
 composition, preserving peripheral parametrizations, grading data,
 coefficients, and the determinant-line normalization of the meridional
 relative model.
\end{enumerate}
If these identifications are established, the equivalence and its
algebraic pairing compatibility follow from
Theorem~\ref{thm:resolved-equivalence}.  They are not consequences of
that theorem alone.

There are relevant existing geometric constructions, which also clarify
the limits of a shortcut.  Auroux identifies the bordered algebra with
endomorphisms of generating arc-product Lagrangians in a partially
wrapped Fukaya category \cite[Theorems~1.2--1.3]{AurouxFukaya}.
The same paper explicitly qualifies the scope of the full geometric
module and gluing arguments after its Theorems~1.4--1.5; we do not
invoke them as an unqualified BV--BFV equivalence.
Chekeres--Losev--Mnev--Youmans obtain enhanced Fukaya--Morse
compositions from a gauge-fixed BF effective action
\cite{CLMY}.  Their construction involves gradient-tree compositions
for generating functions on a closed manifold, rather than an automatic
identification of the knot-exterior BFV category with the stopped
Heegaard-Floer boundary category.  Matching these geometric settings,
including stops and boundary conditions, is an additional research
problem.

Dg and Koszul techniques also occur in bordered Floer theory
\cite{LOTMorphisms}.  The resolution and cotangent model above address
the algebraic part of the comparison; the geometric identifications
remain to be established.
For the particular polynomial quantum sector of nonresonant
three-dimensional Abelian BF theory computed in the next section,
the proposed dg Morita identification is obstructed.
Thus an abstract choice of $B_k$ cannot be regarded as a computation
of that field theory's boundary observable algebra.

\section{Quantum BFV reduction on the torus}
\label{sec:quantum-boundary}
The preceding algebraic constructions leave open whether the original
three-dimensional linear BF theory supplies the torus algebra.
In a precisely specified quantum sector, this question can be answered
directly.  We compute the polarized BFV differential, its polynomial
state cohomology, and the polynomial Weyl observable algebra.
The calculation uses the boundary differential of the
\emph{three-dimensional} Abelian theory, rather than an algebra imported
from Floer theory.

\subsection{The polarized state complex}
\label{subsec:actual-charge}
Fix a characteristic-zero field $k$, a holonomy
$z\in k^\times\setminus\{1\}$ and $u=z-1$.
Set $K=k((\hbar))$, so $\hbar$ is invertible.
Use the twisted cellular cochain complex of the closed boundary torus
with meridional holonomy $z$ and longitudinal holonomy $1$:
\[
 D^0=k\xrightarrow{d^0}D^1=k^2\xrightarrow{d^1}D^2=k,\qquad
 d^0(t)=(ut,0),\quad d^1(r,s)=us .
\]
In the $A$-polarization of the three-dimensional Abelian BF theory,
take polynomial functions on $D[1]$ as the finite cellular state
model, after a fixed half-density trivialization.
This is the polynomial version of the standard cellular
$A$-polarized state construction and quantum BFV operator
\cite[Section~7, especially equations (71)--(74)]{CMR}.
The coordinates and ghost degrees are
\[
 \xi\ (1),\qquad x,y\ (0),\qquad \eta\ (-1).
\]
The odd coordinates $\xi,\eta$ are exterior generators; the even
coordinates $x,y$ are polynomial generators.  The normalized
degree-one BFV differential on
\begin{equation}\label{eq:actual-Fock}
 \mathsf S=
 K[x,y]\otimes_K\Lambda_K(\xi,\eta)
\end{equation}
is
\begin{equation}\label{eq:actual-delta}
 \delta\xi=\delta y=0,\qquad
 \delta x=u\xi,\qquad \delta\eta=uy,\qquad
 \delta=u\xi\partial_x+uy\partial_\eta .
\end{equation}
These formulas are the dual coordinate differential of the displayed
cellular complex, with ghost degree $1-r$ for an $r$-cell.

Let $p_f$ be the canonical momentum conjugate to a coordinate $f$, with
$|p_f|=-|f|$.  Choose the algebraic quantization convention
$p_f=\hbar\partial_f$ and graded commutators
$[p_f,f]_{\mathrm{gr}}=\hbar$.
The normal-ordered quantum charge is
\begin{equation}\label{eq:actual-Theta}
 \widehat\Theta=u(\xi p_x+y p_\eta)=\hbar\delta
 \quad\text{on }\mathsf S .
\end{equation}
Conventional factors of $i$ multiply this charge by a nonzero scalar
and do not change the cohomology over $K$.
The operators in its two summands anticommute and square to zero,
so $\widehat\Theta^2=0$; no higher quantum correction is needed for
this linear boundary charge.
The coordinate vector field $\delta$ has zero divergence, so the
constant half-density trivialization introduces no additional divergence
term in this operator.

\subsection{Vacuum cohomology}
\begin{theorem}[Nonresonant polarized vacuum cohomology]
\label{thm:quantum-vacuum}
The polynomial state complex $(\mathsf S,\delta)$ has a strong deformation
retract onto the vacuum line $K\cdot1$ in ghost degree zero.
Writing $\varepsilon$ for evaluation at zero and
$\iota:K\to\mathsf S$ for the inclusion of constants,
there is a homotopy $H$ with
\[
 \delta H+H\delta=1-\iota\varepsilon,\qquad
 H^2=H\iota=\varepsilon H=0.
\]
In particular $H^\bullet(\mathsf S,\delta)=K$.
The same formulas extend continuously to completion in total polynomial
degree.  They also contract the complex with differential
$\widehat\Theta$ over $K$, after replacing $H$ by $\hbar^{-1}H$.
\end{theorem}
\begin{proof}
Define the degree-$-1$ graded derivation
\[
 \kappa=u^{-1}(x\partial_\xi+\eta\partial_y).
\]
On each of the four generators, direct substitution gives
$(\delta\kappa+\kappa\delta)f=f$.
The graded commutator of the two odd derivations is a degree-zero
derivation, so on all polynomials
\[
 \delta\kappa+\kappa\delta=N,
\]
where $N$ counts total polynomial degree, including the exterior
generators.  Both derivations preserve this degree.
On the degree-$n$ homogeneous component define
$H=\kappa/n$ for $n>0$, and $H=0$ on constants.
All denominators are invertible in characteristic zero.
This gives the retract identity.
The identity $\kappa^2=0$ follows on the generators and hence on the
graded polynomial algebra, proving the side conditions.
For the completed algebra the formula is applied componentwise.
Finally, $\widehat\Theta=\hbar\delta$ and $\hbar$ is invertible in $K$.
\end{proof}

\begin{remark}[The formal-parameter restriction]
Over $k[[\hbar]]$, the unnormalized operator $\hbar\delta$ can leave
$\hbar$-torsion in cohomology; the theorem does not discard that torsion
by an unstated division.  It concerns either the normalized BRST
differential $\delta$, or the quantum charge after passage to
$k((\hbar))$.  It also concerns polynomial states and their stated
completion, not arbitrary distributional sectors or added defects.
The disappearance of residual boundary fields is the familiar
cohomological reduction of cellular BF theory
\cite[Section~7.4]{CMR}; the displayed homotopy supplies it explicitly
for this meridional torus.
\end{remark}

\subsection{The quantum observable dg algebra}
Let $\mathcal W$ be the graded polynomial Weyl algebra generated by
the four coordinates and their four momenta, over $K$, with the canonical
graded commutation relations above and all other pairs graded commuting.
The algebra has the degree-one differential
\[
 d_{\mathcal W}=\hbar^{-1}
       [\widehat\Theta,-]_{\mathrm{gr}}.
\]
It restricts to $\delta$ on the coordinates.  On momenta it is
\begin{equation}\label{eq:actual-momentum-differential}
 d_{\mathcal W}p_\xi=up_x,\qquad
 d_{\mathcal W}p_y=-up_\eta,\qquad
 d_{\mathcal W}p_x=d_{\mathcal W}p_\eta=0.
\end{equation}
The signs are fixed by the graded canonical relations; for instance the
anticommutator of $p_\xi$ with $\xi$ is $\hbar$, while
$[p_y,y]=\hbar$.

\begin{theorem}[Vacuum reduction of the polynomial BFV observables]
\label{thm:Weyl-reduction}
The unit inclusion
\[
 K\longrightarrow(\mathcal W,d_{\mathcal W})
\]
is a quasi-isomorphism of dg algebras.
Consequently
$\operatorname{Perf}(\mathcal W)\simeq\operatorname{Perf}(K)$.
\end{theorem}
\begin{proof}
Normal ordering, with coordinates to the left and momenta to the right,
gives the graded PBW vector-space isomorphism
\[
 \sigma:\operatorname{Sym}_K(V\oplus V^\vee)
                \longrightarrow\mathcal W,
\]
where $V$ is the four-dimensional coordinate complex and $V^\vee$
denotes the conjugate momentum complex with the differential
\eqref{eq:actual-momentum-differential}.
This is a vector-space isomorphism, not an algebra isomorphism.
Since $d_{\mathcal W}$ sends coordinates to coordinates and momenta
to momenta, it preserves normal ordering and corresponds under $\sigma$
to the linear derivation on this graded symmetric algebra.
There are no lower-order reordering terms in this differential.

Both summands of $V\oplus V^\vee$ are contractible.
Extend the coordinate contraction from the previous theorem by
\[
 \kappa(p_x)=u^{-1}p_\xi,\qquad
 \kappa(p_\eta)=-u^{-1}p_y,\qquad
 \kappa(p_\xi)=\kappa(p_y)=0.
\]
On symbols, $d\kappa+\kappa d=N$ again counts total PBW degree.
Dividing by this degree on the positive-degree components and
transporting by $\sigma$ contracts the observable complex onto its
constants.  This contraction need not be an algebra derivation.
It proves that $H^\bullet(\mathcal W)=K\cdot[1]$.
The unit inclusion is itself a unital dg-algebra map and induces that
cohomology isomorphism, proving the theorem.
The perfect-module equivalence follows from dg scalar extension as in
Theorem~\ref{thm:resolved-equivalence}.
\end{proof}

\begin{corollary}[Regular formal interactions do not add vacuum classes]
\label{cor:formal-vacuum}
Let $\lambda$ be a formal parameter of degree zero and let
$d_\lambda=d_{\mathcal W}+\Delta$ be a unital
$K[[\lambda]]$-linear dg-algebra differential on
$\mathcal W[[\lambda]]$, where $\Delta$ raises $\lambda$-adic order
and $d_\lambda^2=0$.  Then the unit inclusion
\[
 K[[\lambda]]\longrightarrow
       (\mathcal W[[\lambda]],d_\lambda)
\]
is a quasi-isomorphism.  Thus a regular unital formal perturbation with
this same leading boundary complex still has only vacuum cohomology.
\end{corollary}
\begin{proof}
Extend the contraction of Theorem~\ref{thm:Weyl-reduction}
$\lambda$-adically.  The series in the homological perturbation lemma
converge in this topology, since $\Delta$ raises $\lambda$-adic order.
Unitality gives $\Delta(1)=0$.  Hence the corrected inclusion remains
the inclusion of constants and the transferred differential on
$K[[\lambda]]$ is zero.  The lemma supplies a deformation retract
onto this complex \cite[Section~4]{CMR}.  The inclusion is a unital
dg-algebra map, so the cohomology isomorphism proves the assertion.
\end{proof}

\subsection{A Hochschild obstruction}
Let $A_K$ be the characteristic-zero degree-zero algebraic lift of the
torus table, as defined in Section~\ref{subsec:six-generator}.
It is placed in ghost degree zero with zero differential.
\begin{theorem}[No Morita identification in the computed quantum sector]
\label{thm:quantum-Morita-obstruction}
There is no $K$-linear dg derived Morita equivalence between
$(\mathcal W,d_{\mathcal W})$ and $A_K$, or between $\mathcal W$
and its six-generator resolution $B_K$.
More precisely, the Morita invariants are
\[
 HH_0(\mathcal W)\cong K,\qquad
 HH_0(A_K)\cong K^2.
\]
Here the left side is the degree-zero homology of the
\emph{derived dg Hochschild complex}; it is not the trace quotient
of the underlying Weyl algebra with its differential forgotten.
\end{theorem}
\begin{proof}
By Theorem~\ref{thm:Weyl-reduction},
the dg algebra $\mathcal W$ is quasi-isomorphic to $K$, so its
derived Hochschild homology is that of $K$.
For the degree-zero algebra $A_K$,
$HH_0(A_K)=A_K/[A_K,A_K]$.
Every off-diagonal chord is a commutator with an idempotent.
The remaining chord elements are
\[
 ab=[a,b],\qquad bc=[b,c],
\]
because $ba=cb=0$.  Thus the whole chord ideal $J$ lies in the
commutator subspace.  Conversely, the quotient $A_K/J=K\times K$
is commutative, so every commutator lies in $J$.
It follows that $[A_K,A_K]=J$ and
$HH_0(A_K)=K\iota_0\oplus K\iota_1$.

Hochschild homology is invariant under $K$-linear dg derived Morita
equivalence \cite{KellerHH}.  Its displayed dimensions rule out the
claimed equivalence.  Since $B_K$ is quasi-isomorphic to $A_K$, the
same obstruction applies to $B_K$.
\end{proof}

\begin{remark}[Exact scope of the conclusion]
This result computes a boundary observable algebra of the
\emph{original linear BF model in its fixed nonresonant sector}
and obstructs its proposed identification with the specified torus
module algebra.  Adding only contractible BRST pairs, or replacing this
observable algebra by a quasi-isomorphic resolution, cannot remove the
obstruction.  Corollary~\ref{cor:formal-vacuum} also treats regular
unital formal interactions with the same leading complex.
The theorem does not exclude other interacting theories,
resonant holonomy, an all-character family, new boundary defects,
stopped or wrapped brane categories, or a different quantization sector.
It does not assert a theorem about changing Floer coefficients:
the torus algebra in this comparison is the explicitly stated
characteristic-zero algebraic lift.
\end{remark}

\subsection{Reduced sewing}
The reduced polarized state space in this sector is a line.
Any BFV-compatible sewing pairing that descends to these reduced spaces
is therefore a bilinear map $K\otimes_K K\to K$, determined by a scalar.
With compatible vacuum normalizations its action is ordinary scalar
multiplication; determinants, torsion factors and density conventions
may change that scalar but cannot create additional reduced boundary
states.  The standard Gaussian cellular sewing construction and
cohomological reduction are described in
\cite[Sections~7.3--7.4]{CMR}.

When the bulk residual fields are also absent, the modified quantum
master equation makes a bulk boundary state a BFV cocycle.
Theorem~\ref{thm:quantum-vacuum} then reduces that state to a scalar
multiple of the vacuum.  Its scalar determinant-line information is
the appropriate location for the Alexander torsion comparison already
proved in this manuscript.  The quantum calculation does not promote
that scalar to a bordered Floer module with chord operations.
Consequently a geometric equivalence cannot be obtained
through the computed polynomial Weyl sector by gauge fixing or dg
resolution alone.  A proposed enriched theory must change the actual
boundary category and compute its new algebra and manifold modules
independently.

\section{Discussion}
The relative meridional torsion identifies the Alexander polynomial
within a based cellular BF model.  Exact prime-cyclotomic evaluation,
combined with a genus bound, determines its coefficients through field
traces; the threshold is sharp.  Field norms recover branched-cover
homology orders, and annular gluing recovers connected-sum
multiplicativity.  For positive L-space knots, the established staircase
classification turns this arithmetic reconstruction into a filtered
Floer model.

Scalar torsion does not determine knot Floer homology for general
knots.  At the boundary-algebra level, the degree-zero torus model also
obstructs single-field tensor factorizations.  Retaining both idempotents
makes reduction conservative on perfect modules, but does not preserve
pairings: the chord ideal gives the precise derived defect and its
three-layer filtration.

The dg resolution in Section~\ref{sec:constructive} retains the chord
operations and their relation and overlap homotopies.  It preserves
perfect-module categories and derived pairings, and its representation
differential has a classical one-dimensional BV--BFV realization.
This construction starts from the torus algebra.  It does not derive
that algebra, or the modules associated to manifolds, from a
three-dimensional field theory.

In the fixed nonresonant polynomial sector of
Section~\ref{sec:quantum-boundary}, over characteristic zero with
$\hbar$ inverted, the quantum BFV states and Weyl observables contract
to the coefficient field.  Their Hochschild invariant obstructs a dg
Morita equivalence with the degree-zero torus algebra.  This conclusion
is specific to the stated sector and coefficient conventions.  A full
geometric BV--BFV/Floer comparison would require an independently
constructed boundary category, compatible manifold modules, and sewing
maps respecting gradings and peripheral parametrizations.

\appendix
\section{Explicit torsion matrices for the trefoil}
\label{app:trefoil}
We verify the absolute/relative normalization by a direct Fox
calculation.  For the right-handed trefoil take the presentation
$\pi_1(X)=\langle x,y\mid xyx=yxy\rangle$ and relator
$r=xyxy^{-1}x^{-1}y^{-1}$.  Both generators map to $t$ under
abelianization.  Fox's rule
$\partial(uv)/\partial x=
\partial u/\partial x+u\,\partial v/\partial x$
and $\partial x^{-1}/\partial x=-x^{-1}$ gives
\begin{equation}\label{eq:trefoilfox}
 \phi\Bigl(\frac{\partial r}{\partial x}\Bigr)
 =1+t^2-t=t^2-t+1,\qquad
 \phi\Bigl(\frac{\partial r}{\partial y}\Bigr)
 =-(t^2-t+1).
\end{equation}
The cellular chain complex (with the preferred basis of generators)
is
\[
 0\longrightarrow F\xrightarrow{\binom{q}{-q}}F^2
 \xrightarrow{(t-1\quad t-1)}F\longrightarrow0,
 \qquad q=t^2-t+1.
\]
Its torsion is $q/(t-1)$, up to a conventional sign and monomial.
The meridional annulus contributes $1/(t-1)$, leaving the relative
torsion class $q$.  Multiplying by $t^{-1}$ yields the unique
symmetric normalized representative
$\mathscr B_{3_1}(t)=t^{-1}-1+t$.
There are no singular prime-order specializations, since
$q(\zeta_p)\ne0$ for every prime $p$.

\section{Acyclic Gaussian determinants and the scope of BV pushforward}
\label{app:gaussian}
For completeness, let $(C^\bullet,d)$ be a finite-dimensional acyclic
complex with fixed bases and a decomposition
\[
 C^k=B^k\oplus L^k,\qquad B^k=d(C^{k-1}),\quad
 d:L^k\overset{\sim}{\longrightarrow}B^{k+1}.
\]
In a basis adapted to the direct sums, the torsion is an alternating
product of determinants of the isomorphisms
$d|_{L^k}:L^k\to B^{k+1}$, together with the determinant of the
basis changes.  In the Milnor convention used in this paper,
this product is fixed by the requirement that the two-term
circle complex \eqref{eq:circle-complex} has torsion $(z-1)^{-1}$.
The mixed relative cochain complex carries an analogous splitting,
and its determinant-line element is independent of the splitting
as an algebraic torsion, though it still depends on cellular lifts.

Finite-dimensional Abelian BV pushforward is a Gaussian/Berezin
integration after choosing a Lagrangian gauge-fixing subspace.
Its determinant factors are precisely such alternating products;
an actual state also involves half-densities and possibly
universal phases.  For a non-unimodular $L_z$, no invariant complex
number is assigned without controlling Euler structures or bases.
Our \emph{BF character} is therefore explicitly the normalized
Milnor-torsion branch of the finite-dimensional BF determinant-line
output.  None of the present proofs requires selecting a functional
integral measure or claiming an equality of unnormalized path
integrals.

\section*{Use of generative AI}
The manuscript was prepared in \LaTeX{} and compiled with pdf-\LaTeX{}.
Generative AI tools assisted with the formulation and editing of the prose,
LaTeX typesetting, exploratory calculations, the development of mathematical
arguments, and the identification of relevant references. The author takes
responsibility for the content of the manuscript.

\end{document}